\documentclass[12pt]{article}
\usepackage[]{geometry}
\usepackage{graphicx}
\usepackage{amsmath}
\usepackage{amsthm, thmtools}
\usepackage{amsfonts}
\usepackage{amssymb}
\usepackage{booktabs}
\usepackage{nicematrix}
\usepackage{color}
\usepackage{authblk}
\usepackage{verbatim}
\usepackage{hyperref}
\usepackage{ulem}
\usepackage{bm}
\usepackage{nicematrix}
\usepackage{tikz}
\usepackage{algorithm}
\usepackage[noend]{algpseudocode}
\usepackage{lineno}

\declaretheoremstyle[
  headfont=\bfseries,
  bodyfont=\itshape,
  spaceabove=\topsep,
  spacebelow=\topsep,
]{custom-plain}
\declaretheoremstyle[
  headfont=\itshape,
  bodyfont=\normalfont,
]{custom-remark}
\declaretheoremstyle[
  headfont=\bfseries,
  bodyfont=\normalfont,
]{custom-definition}

\declaretheorem[
  style=custom-plain,
  name=Lemma,
]{lem}
\declaretheorem[
  style=custom-plain,
  name=Proposition,
]{prop}

\declaretheorem[
  style=custom-remark,
  name=Remark,
]{rmk}
\declaretheorem[
  style=custom-remark,
  name=Example,
]{example}

\declaretheorem[
  style=custom-definition,
  name=Definition,
]{defn}

\numberwithin{equation}{section}

\def\N{\mathbb{N}}

\def\R{\mathbb{R}}

\def\1{\mathbf{1}}

\DeclareMathOperator{\Var}{Var}

\providecommand{\keywords}[1]
{
  \small	
  \textbf{\textit{Keywords---}} #1
}

\title{Generalized Measures of Population Synchrony}
\date{}

\author[a,*]{Francis C. Motta}

\author[b]{Kevin McGoff}

\author[c]{Breschine Cummins}

\author[d]{Steven B. Haase}

\affil[a]{Department of Mathematics and Statistics, Florida Atlantic University, Boca Raton, FL, USA}
\affil[b]{Department of Mathematics and Statistics, UNC Charlotte, Charlotte, NC, USA}
\affil[c]{Department of Mathematical Sciences, Montana State University, Bozeman, MT, USA}
\affil[d]{Department of Biology, Duke University, Durham, NC, USA}
\affil[*]{Corresponding author: Francis C. Motta, fmotta@fau.edu}

\begin{document}
\maketitle


\begin{abstract}
Synchronized behavior among individuals, broadly defined, is a ubiquitous feature of populations. Understanding mechanisms of (de)synchronization demands meaningful, interpretable, computable quantifications of synchrony, relevant to measurements that can be made of complex, dynamic populations. Despite the importance to analyzing and modeling populations, existing notions of synchrony often lack rigorous definitions, may be specialized to a particular experimental system and/or measurement, or may have undesirable properties that limit their utility. Here we introduce a notion of synchrony for populations of individuals occupying a compact metric space that depends on the Fr\'{e}chet variance of the distribution of individuals across the space. We establish several fundamental and desirable mathematical properties of our proposed measure of synchrony, including continuity and invariance to metric scaling. We establish a general approximation result that controls the disparity between synchrony in the true space and the synchrony observed through a discretization of state space, as may occur when observable states are limited by measurement constraints. We develop efficient algorithms to compute synchrony for distributions in a variety of state spaces, including all finite state spaces and empirical distributions on the circle, and provide accessible implementations in an open-source Python module. To demonstrate the usefulness of the synchrony measure in biological applications, we investigate several biologically relevant models of mechanisms that can alter the dynamics of population synchrony over time, and reanalyze published experimental and model data concerning the dynamics of the intraerythrocytic developmental cycles of \textit{Plasmodium} parasites. We anticipate that the rigorous definition of population synchrony and the mathematical and biological results presented here will be broadly useful in analyzing and modeling populations in a variety of contexts.
\end{abstract}

\keywords{Population Synchrony, Fr\'{e}chet Variance, Wasserstein Distance,  Optimal Transport, Biological Clocks, \textit{Plasmodium}}

\maketitle



\section{Introduction}

Individuals within a population, whether particles, cells, insects, or humans, often exhibit some form of synchrony, and this may be advantageous. For example, synchronizing processes is essential for coordinating functions of cell populations.  In metazoans, cells of a particular organ or tissue must act together to respond to the dynamic needs of the organism.  Synchronization is also observed in single cell organisms when populations act coordinately to form structures such as biofilms \cite{RN1913, RN2212}, or to promote mating in the case of yeast cells that synchronize in G1 phase of the cell-cycle in response to mating pheromone \cite{RN2238, RN2237, RN732}.  Cells also tend to synchronize rhythmic behaviors such as circadian cycles \cite{RN2009}.  In metazoans autonomous cellular clocks synchronize within a tissue and with central clocks of the suprachiasmatic nucleus \cite{RN2229, RN2223, RN2226, RN2228, RN2232}.  In turn, the central clock synchronizes with the 24-hour period of the earth’s rotation.  Single cell organisms also exhibit circadian behavior and can synchronize to daily cycles via environmental signals such as light/dark or temperature cycles \cite{RN2236, RN2235, RN2233, RN2234}.  Relatively recent work has also shown that single-cell parasites can synchronize their developmental cycles with their host circadian cycle \cite{RN1728, Motta2023-pnas, science}, although the mechanisms are not yet understood. Additionally, synchrony of neural cells is believed to play an important role in brain function \cite{neuro2011,uhlhaas2009neural}.

Mathematical modeling is one approach to address the mechanistic questions of how cells synchronize various processes and/or entrain to environmental rhythms.  Such approaches will require rigorously defined mathematical definitions of synchrony (and asynchrony) that faithfully capture the behaviors of the synchronizing system.  Multiple definitions have been proposed, but tend not to be generalizable, or do not accurately represent important features of the system (e.g., variance).  Here we present a rigorous definition of synchrony and a generalized framework for modeling synchrony/asynchrony that should aid biologists (and others) in understanding the regulation of synchrony and entrainment processes.



The terms synchrony and synchronization appear in a variety of settings in the scientific literature. Rigorous measures of synchrony include the famous Kuramoto order parameter, developed to interrogate the Kuromoto oscillator model \cite{Kuramoto1975}, and computable from a finite collection of phases in the circle encoded as complex numbers with modulus 1. The order parameter is given by the modulus of the average of these complex numbers and itself will be 1 only if all phases are identical, indicating perfect phase coherence \cite{STROGATZ20001}. Another quantification is the synchrony measure $\chi \in [0,1]$, representing a ratio of a time-averaged variance of a population mean to the average of the time-averaged variances \cite{Golomb1994-fr, Golomb1993-ca, Tayarpnassynch}, used to model synchrony across networked compartments (e.g., cells).  In ecology, correlation or cross-correlation of time series data for population dynamics, derivatives of population dynamics, or phase is often used to assess synchrony between populations \cite{liebholdsyncheco}. In \cite{Ottesenpnassynch}, the synchrony of time series transcriptomics data between three marine microbial species was assessed by several heuristic measures: (1) computing correlation between the square root of the sum of squared differences in abundance for all transcripts at each time point, (2) computing correlation of the relative abundance of subsets of functionally related genes across time, and (3) performing Procrustes analysis (shape similarity) on the clustering patterns of the projected time series data onto first two principal components for each species.  More recently, a measure of population synchrony has been proposed for populations of \textit{Plasmodium} parasites and depends on the time taken for most of the population to traverse its intraerythrocytic cycle (IEC) \cite{Greischar2024}.  In studies of the circadian rhythm, the standard deviation of the distribution of cellular periods has also been referred to as synchrony \cite{HERZOG20153}. Finally, although not usually referred to as synchrony, in economics, the Gini coefficient (index) is a common quantification of income inequality ranging from 0 (perfect equality) to 1 (if all wealth concentrated in a single individual) \cite{gini1921}. This (and other quantifications of income inequality \cite{DeMaio2007}) may be seen to measure how synchronized a population is in the set of possible values of wealth, where perfect equality could be interpreted as perfect synchrony.


Given this vast body of work, why is another definition of synchrony useful? For one, existing measures of synchrony are specialized to specific systems and experimental modalities; e.g., the Kuramoto order parameter is appropriate only for a population of oscillators, but not for non-cyclic processes such as developmental progression. Additionally, a common property in the previously listed definitions of synchrony is that they depend on the variability of specific population dynamics, or worse, convolve population distributions across states and changes over time in population distributions across states. Such properties may reduce the utility of the measurement in modeling and analyzing dynamical processes that exhibit multiple (possibly unknown) synchronizing/de-synchronizing mechanisms.

We aim to provide a definition that is broadly applicable to many systems, models, and measurements, and that permits rigorous quantitative comparison of synchrony across experiments and time, with desirable properties that are rigorously established.  With our proposed measure, we hope to clarify issues that have arisen in the literature, e.g., the impact of replication and discrete or categorical measurements on intuitive ideas of time-dependent synchrony.

We take the view that the synchrony of a population should be a measure of the extent to which the individuals in the population occupy the same state at the same time. As such, we propose a precise mathematical function that takes the distribution of states of the population at one instant of time and returns a number between $0$ and $1$, referred to as the synchrony of the population at that time, where $1$ indicates perfect synchrony (all individuals occupy exactly the same state) and $0$ indicates that the population is as far from being perfectly synchronized as is possible. 

The proposed definition of synchrony, stated precisely in Section~\ref{sec:defining}, is based on the Fr\'{e}chet variance of the population in a metric space $(M,d)$. This definition fulfills five desirable properties, the first four of which are discussed in Section~\ref{sec:properties}. The fifth property, computability of our quantification of synchrony, depends on the choice of $(M,d)$ and is addressed through examples in Section~\ref{sec:examples} and algorithms in Section~\ref{sec:computing_asynch}. We follow with several applications of the methodology in Section~\ref{sec:results} to measurements and models of biological systems, and discuss some of the mathematical contributions of this work and the implications of this notion of synchrony to experimental sciences in Section~\ref{sec:discuss}.

\section{Defining Synchrony}
\label{sec:defining}
Consider a population of individuals distributed within some fixed state space. For example, yeast cells progressing through their cell-division cycle or malaria parasites going through the IEC may be thought of as occupying phases in the circle $S^1 \cong [0,1)$, representing the fraction of their periodic developmental cycles. A population of organisms dispersed within a habitat may be regarded as a distribution over space, while a collection of mammalian cells of different cell types within a multicelluar organism might be thought of as individuals sampled from a gene expression space. 

We suppose that at any given instant of time, each individual is in some measurable state, and we would like to quantify the amount of synchronization between individuals within this population with respect to these states. 
Together these individuals form a population distribution on the state space. A natural idea of perfect synchrony is that this population distribution has zero variance; i.e., all members of the population occupy the same state. This idea of perfect synchronization is generalizable to arbitrary metric spaces, and we therefore propose to measure the deviation from perfect synchronization using a normalized measure of distance to the nearest delta distribution, which concentrates all mass on a single state. This distance to the nearest point mass then serves as a generalized notion of the variance of the observed population distribution.

Let $(M,d)$ be a compact metric space that contains the set of possible states (i.e., the state space) and let $\mathcal{P}(M)$ be the set of Borel probability measures on $M$ endowed with the weak$^*$ topology.  At any given instant of time, the distribution of the population across the state space $M$ can be represented by a probability distribution $\pi \in \mathcal{P}(M)$. For each point $a \in M$ there is a corresponding delta distribution, $\delta_a \in \mathcal{P}(M)$, defined for any measurable $A\subset M$ by $\delta_a(A) = 1$ if $a\in A$ and 0 otherwise. To quantify synchrony, we seek a function $F : \mathcal{P}(M) \to [0,1]$, where $F(\pi)$ is interpreted as the extent to which the distribution $\pi$ deviates from a $\delta_a$, $a\in M$. 

We suggest that the synchrony function, $F$, should have the following properties: 
\begin{enumerate}
\renewcommand{\labelenumi}{(\theenumi)}
\item $F$ should depend only on the distribution of individuals in a state space.
\item $F$ should be interpretable and reliably detect perfect synchrony, achieved only by delta distributions on $M$. That is, $0 \leq F(\pi) \leq 1$, with $F(\pi)=1$ if and only if $\pi = \delta_x$ for some $x \in M$, and $F(\pi) = 0$ is guaranteed for some distribution representing a maximally asynchronous population in the chosen state space.
\item $F$ should vary continuously with its input distribution.
\item  $F$ should be broadly applicable and useful to experimental data and models involving populations.
\item $F$ should be efficiently computable for a large class of spaces.
\end{enumerate}

\noindent As we will show, the following definition of synchrony enjoys each of the above properties.

\begin{defn} Let $(M,d)$ be a compact metric space containing at least two points and $\pi \in \mathcal{P}(M)$ a probability measure on the Borel $\sigma$-algebra of $(M,d)$.  The \textbf{synchrony} of the distribution $\pi$ is defined to be  
\begin{equation}
\label{eq:synchrony}
\begin{aligned}
F_{(M,d)}(\pi) & = 1-\frac{1}{\nu_{(M,d)}} \Var_{(M,d)}(\pi)^{1/2}  \\ 
               & = 1-\frac{1}{\nu_{(M,d)}}\inf\limits_{\alpha \in M} \; \left( \int_{M} d(x,\alpha)^2 \,d\pi(x) \right)^{1/2} \\
               & = 1-\frac{1}{\nu_{(M,d)}} \inf\limits_{\alpha \in M} \; W_2(\pi, \delta_\alpha),
\end{aligned}
\end{equation}
\noindent where the normalization factor $\nu_{(M,d)}$ is given by
\begin{equation}
\begin{aligned}
\nu_{(M,d)} & := \sup_{\pi' \in \mathcal{P}(M)} \; \Var_{(M,d)}(\pi')^{1/2} \\
            & =  \sup_{\pi' \in \mathcal{P}(M)} \; \inf\limits_{\alpha \in M} \; W_2(\pi', \delta_\alpha),
\end{aligned}
\label{eq:norm_const}
\end{equation}
and $W_2$ is the Wasserstein-2 distance from optimal transport \cite{KantorovichW2metric}. We note that if $M$ contains at least two points, then $\nu_{(M,d)} >0$, and thus $F_{M,d}$ is well-defined.
 If $(M,d) = (\{x\},0)$ consists of a single point, then so will $\mathcal{P}(M) = \{\delta_x\}$, and so both $\nu_{(\{x\},0)} = 0$ and $\Var_{(\{x\},0)}(\delta_x) = 0$.  In this case we define $F_{(\{x\},0)}(\pi) = 1$. 
 In every case, the quantity 
 \[1-F(\pi) = \Var_{(M,d)}(\pi)^{1/2}/\nu_{(M,d)} = \inf\limits_{\alpha \in M} \; W_2(\pi, \delta_\alpha)/\nu_{(M,d)}\] 
 is regarded as the \textbf{asynchrony} of $\pi$.
\label{def:asynchrony}
\end{defn}

We may suppress the subscript $(M,d)$ in our notation (i.e.,  $F(\pi) = F_{(M,d)}(\pi)$, $\Var(\pi) = \Var_{(M,d)}(\pi)$, and $\nu = \nu_{(M,d)}$) when explicit reference to the metric space is unnecessary. The term $\Var(\pi)$ is known as the Fr\'{e}chet (or generalized) variance of the distribution $\pi$ \cite{Fréchet1948}, and so its square root may be regarded as a kind of generalized standard deviation. Any $\alpha \in M$ that achieves the infimum in Eq.~\eqref{eq:synchrony} is called a Fr\'{e}chet mean, or barycenter of the distribution $\pi$ \cite{Fréchet1948}. Since we assume that $(M,d)$ is compact, the infimum in Eq.~\eqref{eq:synchrony} is realized and so every distribution $\pi \in \mathcal{P}(M)$ has a Fr\'{e}chet mean (although it may not be unique). Furthermore, since $\mathcal{P}(M)$ is also compact \cite[Thm. 15.11, p. 513]{aliprantis06} and $\Var : \mathcal{P}(M) \to \R$ is continuous, the supremum in Eq.~\eqref{eq:norm_const} is also achieved. This normalization factor, $\nu$, is the square root of the maximum generalized variance of any measure supported on $(M,d)$, which we think of as corresponding to a maximally asynchronous population.

As indicated in (\ref{eq:synchrony}), our choice of synchrony can be written in terms of the Wasserstein-2 distance from optimal transport \cite{KantorovichW2metric,villani2009optimal}. In this interpretation, the asynchrony of a population distribution, $1-F(\pi)$, reflects the minimal distance between $\pi$ and any perfectly synchronized population. 

\section{Properties of Synchrony}
\label{sec:properties}
\subsection{Continuity}
\label{ssec:continuity}
By construction, the synchrony function given in Eq.~\eqref{eq:synchrony} depends only on how a population is distributed in its state space, and so property (1) is satisfied. Moreover, the following proposition gives that our notion of synchrony varies continuously with the underlying distribution (property (3)) whenever $(M,d)$ is compact. 

\begin{restatable}{prop}{continuity}
\label{prop:continuity}
Assume $(M,d)$ is a compact metric space and $\mathcal{P}(M)$ is endowed with the weak$^*$ topology. Then $F : \mathcal{P}(M) \to \R$ is continuous. 
\end{restatable}

For a proof of Proposition \ref{prop:continuity}, see Appendix \ref{app:proof_continuity}. 
This proposition has several interpretations and implications. First, at a conceptual level, it states that small changes in the population will result in correspondingly small changes to the synchrony measure $F$. Furthermore, if the population distribution evolves continuously over time, then the resulting synchrony measure $F(\pi_t)$ will be a continuous function of time, and thus the synchrony measure $F$ can be meaningfully measured over a population whose distribution changes continuously in time. Lastly, the continuity of $F$ guarantees a type of consistency property with respect to statistical estimation and measurement noise: if $\{\hat{\pi}_n\}_{n=1}^{\infty}$ is a sequence of estimates of $\pi$ that converges to $\pi$ in $\mathcal{P}(M)$,  e.g., resulting from finite samples of the population with potentially noisy observations, then the synchrony measure satisfies $F(\hat{\pi}_n) \rightarrow F(\pi)$. This provides reassurance that the synchrony of a large finite sample approximates the synchrony of the full population.
\subsection{Interpretability}
\label{ssec:interpretability}
As a consequence of the continuity of $\Var$ and compactness of $\mathcal{P}(M)$, the generalized variances of distributions over $\mathcal{P}(M)$ are bounded and a there exists a maximizer $\pi^* \in \mathcal{P}(M)$ in  Eq.~\eqref{eq:norm_const}. We regard such a distribution on $M$ as maximally asynchronous, and thus $F$ satisfies property (2) and provides an interpretable measure of synchrony, where $F(\pi) = 1$ if and only if the population is perfectly synchronized in some state, i.e., $\pi = \delta_{\alpha}$ for some $\alpha \in M$. This follows from the elementary fact that the only distributions with 0 generalized variance are delta distributions supported on a single state:
\begin{restatable}{prop}{vardelta}
\label{prop:vardelta}
Let $(M,d)$ be a compact metric space, and $\pi \in \mathcal{P}(M)$. Then $\Var(\pi)=0$ if and only if $\pi = \delta_{\alpha}$ is a delta distribution supported at some $\alpha \in M$.
\end{restatable}
\noindent For completeness we include a proof of this fact in Appendix~\ref{app:vardelta}. 

The normalization of $F$ to the interval $[0,1]$ permits quantitative comparison of synchrony across observations of systems under differing conditions. It also ensures that the measure is invariant to rescaling the metric, enabling, for example, comparison of periodic processes with different periods.

\subsection{Invariance to Rescaling}
\label{ssec:scaleinvariant}
The measure of synchrony given in Eq.~\eqref{eq:synchrony} depends explicitly on the choice of state space and metric on it. While in practice we expect the states to be determined by the available measurements, the metric on those states may be subject to greater choice. For example, if the state space is a circle, then one may choose to parameterize it by radians in $[0,2\pi)$, resulting in distances between points in $[0,\pi]$, or by fractions of a cycle, i.e., by $[0,1)$, resulting in distances between points being confined to $[0,1/2]$. Such a linear reparameterization should not change the synchrony of a distribution, and indeed it does not. This can be seen as a special case of rescaling a metric on a common space, which leaves synchrony unchanged:
\begin{restatable}{prop}{metricrescale}
\label{prop:metric_rescale}
Let $(M, d)$ be a metric space. For $\lambda >0$,  define $d_\lambda(x,y) = \lambda d(x,y)$ for all $x,y\in M$. Then $F_{(M,d)}(\pi) = F_{(M,d_\lambda)}(\pi)$ for all $\pi \in \mathcal{P}(M)$. 
\end{restatable} 
\noindent An elementary proof of the fact that synchrony is invariant to rescaling the metric is provided in Appendix~\ref{app:metric_rescale}.

\subsection{Approximability and Quantization of States}
\label{ssec:Approximability}
In practice, it might be impossible to measure the true states of a system exactly. Instead, due to experimental capabilities or other factors, it may be that the true states are only \textit{approximated} by a measurement. For instance, a state space may be discretized into a finite number of discrete states (e.g., a spatial grid), or an individual's state may be represented by an observable proxy (e.g., cell cycle stage coarsely characterized by cellular morphology). In this section we suppose that $(M,d)$ is the true underlying state space of a population distributed according to $\pi \in \mathcal{P}(M)$, and furthermore the population is only observed through a measurable function $g : M \to X \subset M$, where $X$ is the set of observable states.  In other words, if an individual is in state $x \in M$, then when we measure that individual our measurement/experiment  returns only the value $g(x) \in X$. Although not required, in practice we often expect $X$ to be finite. 

If we observe a population using the function $g : M \to X$, then the distribution of the observations is given by the push-forward measure $\pi_X := g_* \pi = \pi \circ g^{-1}$, defined by $\pi_X(E) := \pi(g^{-1}(E))$ for all measurable sets $E \subset X$. If $(X,d)$ is regarded as a subspace of $(M,d)$ then we can assess the synchrony of the population in the observed state space, $ F_{(X,d)}(\pi_X)$, although we do not have access to the true population synchrony $ F_{(M,d)}(\pi)$. To what extent does the observable synchrony in $X$ approximate the true underlying synchrony in $M$? Not surprisingly, the answer depends on how well $X$ approximates $M$, a fact which is made precise in the following proposition.

\begin{restatable}{prop}{approximability}
\label{prop:approximability}
 Let $(M,d)$ be a compact metric space with at least two points, $X \subset M$ compact, and $g: M \to X$ measurable. Let $\epsilon \geq 0$ and assume $d(x,g(x)) \leq \epsilon$ for all $x\in M$. Let $\pi \in \mathcal{P}(M)$ and take $\pi_X := g_*\pi \in \mathcal{P}(X)$. Then
    \begin{equation*}
        \left| F_{(M,d)}(\pi) - F_{(X,d)}(\pi_X) \right| \leq \frac{3\epsilon}{\nu_{(M,d)}},
    \end{equation*}
    \noindent where $\nu_{(M,d)}$ is the synchrony normalization constant of $(M,d)$.
\end{restatable}

A proof of Proposition~\ref{prop:approximability} appears in Appendix~\ref{app:approximability}. The function $g$ in Proposition~\ref{prop:approximability} serves a dual role. First, it encodes the relationship between the true states and the observed states, indicating the measured state $g(x)$ corresponding to each true state $x$. Second, it satisfies a bound on the largest distance between any state and its observed state, here controlled by $\epsilon$. Thus Proposition~\ref{prop:approximability} provides precise control on how far away an observed measure of synchrony can be from the true population synchrony in the unobserved state space, provided that one knows how well $X$ approximates $M$. Another consequence of Proposition~\ref{prop:approximability} is that the true synchrony of a population can be approximated arbitrarily well by a judicious choice of subspace that well-approximates $M$, in the sense that each point $x \in M$ is within a distance $\epsilon$ of at least one measurable state, namely $g(x) \in X$. 

Although the statement of Proposition~\ref{prop:approximability} does not specify it, nor does the proof rely on it, we imagine that in practice $X$ will typically be equal to a finite set $\{x_1, \ldots, x_n\} \subset M$, induced by discretization of the true state space, and the function $g: M \to X$ will send each state to its nearest observable state, i.e., $g(x) = x_i$ if $d(x,x_i) < d(x,x_j)$ for all $j\neq i$. For those states that are equally far from multiple observable states, some choice of a tie-breaking rule would need to be made to ensure $g$ is a measurable function. However, the conclusion of Proposition~\ref{prop:approximability} does not depend on this choice. 

Because $F$ can be well-approximated both by incomplete observations about the population (i.e., through empirical samples of the underlying population distribution, due to continuity) and incomplete observations of the true states of the system, we expect this notion of synchrony will be broadly useful in the context of experimental data which may involve uncertainties in both facets, i.e., it satisfies property (4).
 
\section{Example State Spaces}
\label{sec:examples}
In this section we contextualize the definition of population synchrony and its mathematical properties by describing several examples which may be particularly relevant in applications. Then in Section~\ref{sec:computing_asynch} we establish results that allow efficient computation of synchrony in several important spaces (property (5)). 

To compute the synchrony of any distribution $\pi \in \mathcal{P}(M)$ according to Eq.~\eqref{eq:synchrony}, one must determine the maximum asynchrony of any distribution on $M$. This essentially reduces to maximizing the generalized variance, and so we focus first on example spaces whose maximal generalized variance is known.

\begin{example}[Compact Intervals]
If $M = [a,b] \subset \mathbb{R}$, with the usual metric, $d(x,y)=|x-y|$, then  
\[\Var(\pi) \leq \frac{1}{4}(b-a)^2, \]
with equality only at $\pi^* = \frac{1}{2}(\delta_a + \delta_b)$ \cite{popoviciu1935equations}.

 Notably, the uniform distribution is not maximally asynchronous over an interval. A straightforward calculation shows that the uniform distribution on $[a,b]$, $\pi$, has generalized variance $\Var(\pi) = (b-a)^2/12$. Therefore the synchrony of the uniform distribution is $F(\pi) = (\sqrt{3}-1)/3$, independent of the particular interval $[a,b]$. This can again be seen as a special case of Proposition~\ref{prop:metric_rescale}, where, in this case, normalization ensures synchrony is independent of the scale of the interval.
\label{exp:compact_interval}
\end{example}



%
 %
    
\begin{example}[Compact Subsets of $\mathbb{R}^n$]
\label{exp:compact_subset_Rn}
Recent results generalize Example~\ref{exp:compact_interval} to subsets of $\mathbb{R}^n$ with the usual metric \cite{LimMcCann2022} and more general metric spaces \cite{Pass2022}. Even in the Euclidean setting, determination of the maximal generalized variance supported on a compact subset needs to be handled on a case by case basis, with notable results for some special cases (see \cite[Exp. 1.5]{LimMcCann2022}). For instance, generalizing Example \ref{exp:compact_interval} slightly, for $M = [a,b]^n$ with the Euclidean metric, the maximal variance will be $n(b-a)^2/4$, achieved by concentrating equal mass on the $2^n$ corners of $M$. 


\end{example}

\begin{example}[The Circle]
Suppose $M = [0,1)\cong S^1$ is the circle with circumference 1, with arclength metric 
\begin{equation*}
d(x,y) = \min( |x-y|, 1-|x - y|). 
\end{equation*} Then the following fact appears as a special case in \cite[Exp. 9]{Pass2022}.
\begin{prop}
\label{prop:circle_var_max}
Suppose $(S^{1},d)$ is the circle and $\pi^*$ is the uniform distribution (Lebesgue measure) on $S^{1}$. Then for all $\pi \in \mathcal{P}(S^{1})$, we have
\begin{equation*}
\Var\;(\pi) \leq \Var\;(\pi^*),
\end{equation*}
with equality if and only if $\pi = \pi^*$. 
\end{prop}
In the present setting of population synchrony, Proposition \ref{prop:circle_var_max} establishes that the most asynchronous population on a circle is one which is uniformly distributed. A calculation then shows that the maximum generalized variance of any distribution on the circle of circumference one is $\int_{S^1} d(x,0)^2 \; d\pi(x) = 1/12$, and so $\nu_{(S^1,d)} = \sqrt{3}/6$.  
\end{example}

\begin{example}[Discrete Circles (Uniform)]
\label{exp:discrete_circle_uniform}
Let $S^{1}_p = \{0,\dots,p-1\}$, for some $p \in \N$, and define 
\begin{equation*}
d_p(x,y) = \min(|x-y|/p, 1 - |x-y|/p),
\end{equation*}
for each $x,y \in S^{1}_p$.
We think of $(S^{1}_p,d_p)$ as a discrete circle with $p$ states. An analogous result to Proposition \ref{prop:circle_var_max} holds for discrete circles.
\begin{restatable}{prop}{discretecircle}
\label{prop:finite_cyclic_var_max}
Suppose $(S^{1}_p,d_p)$ is a discrete circle, and let $\pi^*$ be the uniform distribution on $S^{1}_p$. Then for all $\pi \in \mathcal{P}(S^{1}_p)$
\begin{equation*}
\Var\;(\pi) \leq \Var\;(\pi^*).
\end{equation*}
\end{restatable}
\noindent As a consequence of Proposition~\ref{prop:finite_cyclic_var_max} (proved in Appendix~\ref{app:finite_cyclic_var_max}) it follows that,
$$
\begin{aligned}
\nu_{(S^{1}_p,d_p)}^2 & =  \frac{1}{p}\sum _{i=0}^{p-1} \min \{i/p,1-i/p\}^2 \\
  & = \frac{6 \left\lceil \frac{p}{2}\right\rceil^2-6 p \left\lceil \frac{p}{2}\right\rceil-6 \left\lceil \frac{p}{2}\right\rceil +2 p^2+3 p+1}{6 p^2}.
\end{aligned}
$$

We may regard $(S^1_p, d_p)$ as a discretization of $S^1$ into $p$ states representing arcs of length $1/p$, with the distances between states the arclengths between the midpoints of each pair of arcs. In other words, we approximate the circle by a subset of $p$ equally-spaced points, $x_i = i/p$, $i=0,\ldots, p-1$.  If $g: S^1 \to S^1_p$ assigns elements in each (half-open) arc to its midpoint, i.e.,  $g([x_i-1/p,x_i+1/p)) = \{x_i\}$, (and points equally-far from $x_i$ and $x_{i+1}$ to one or the other), then $d(x,g(x)) \leq 1/2p$ for all $x\in S^1$. Therefore, Proposition~\ref{prop:approximability} yields 
$$
\left| F_{(S^1,d)}(\pi) - F_{(S^1_p,d_p)}(g_* \pi) \right| < \frac{3\sqrt{3}}{p}.
$$
Such an approximation may occur when a periodic process, whose true state space may be modeled as phases on a circle, is replaced by an approximation into finitely many discrete \textit{stages} representing intervals of phases. This is often done in experimental biology studies involving periodic phenomenon (e.g., cell division cycle, developmental cycles, or circadian cycles) \cite{RN1840, RN1452, science}. 
\end{example}

\begin{example}[Discrete Circles (General)]
\label{exp:discrete_circle_general}
In applications it may be more accurate to model states as arcs with \textit{different} ``stage-dependent'' lengths, still taking the distance between states as the arc length between the corresponding midpoints of the arcs they represent. 
Suppose $M = \{0,\ldots, p-1\}$ for some $p \in \mathbb{N}$, and partition $[0,1)$ into disjoint arcs $[a_i,a_{i+1}), i=0,\ldots,p-1$, where $a_0=0$ and $a_{p}=1$. Then the distance between states $x,y$ is
\begin{equation}
d(x,y) = \min(|a_x - a_y+a_{x+1}-a_{y+1}|, 2-|a_x - a_y + a_{x+1}-a_{y+1}|)/2.
\label{eq:discrete_circle_general}
\end{equation}
Such spaces may be appropriate to model periodic developmental cycles represented by discrete morphological stages that occupy different fractions of the cycle, as observed through microscopy measurements of cell populations \cite{science}. Unlike the uniform discrete circle spaces defined in Example~\ref{exp:discrete_circle_uniform}, the distribution with maximal asynchrony over the space given by Eq.~\eqref{eq:discrete_circle_general} need not be uniform. For instance, choosing $a_0=0, a_1=1/2, a_2=3/4, a_3=1$ defines a three-state circular space with one long state occupying half the cycle and two shorter states, each occupying 1/4 of the cycle. In this case, the uniform distribution $\pi=[1/3, 1/3, 1/3]$ has generalized variance $\Var(\pi) = 1/3((3/8)^2 + (1/4)^2)$ and synchrony $F(\pi) = 1 - \sqrt{1/3((3/8)^2 + (1/4)^2}/(9/32) \approx 0.0748$. On the other hand, a distribution with maximal generalized variance (equal to 81/1024) on this space is $\pi^* = [7/16, 9/32, 9/32]$. The uniform distribution concentrates more mass in each of the ``shorter'' states and less mass in the ``longer'' state compared with $\pi^*$. In this way the uniform distribution is somewhat more synchronized than $\pi^*$.

As in Example~\ref{exp:discrete_circle_uniform}, Proposition~\ref{prop:approximability} can be used to provide an upper bound on the error that may be incurred by imprecise representation of states occupying each arc, $[a_i, a_{i+1})$, $i=0,\ldots, p-1$.
\end{example}

\section{Computing Synchrony}
\label{sec:computing_asynch}

As can be seen from the preceding examples, computing the synchrony normalization constant must be handled for each space individually, often without knowing \textit{a priori} which distribution(s) will have maximal asynchrony. Both the determination of $\nu_{(M,d)}$ and the calculation of the synchrony for a given measure $\pi$ requires determination of its generalized variance.  As observed, this can be computed by finding the nearest delta distribution to $\pi$ in the Wasserstein-2 metric. 

In general, if $(M,d)$ is a complete and separable metric space (which will be the case for $M$ compact) one defines the Wasserstein-2 distance between probability measures (with finite variance) $\pi,\mu \in P(M)$ by
$$
W_{2}(\pi,\mu) = \left(\inf_{\gamma \in \Gamma[\pi,\mu]} \int_{M\times M} d(x,y)^2 \; d\gamma(x,y)\right)^{1/2},
$$
\noindent where $\gamma \in \Gamma[\pi,\mu]$ is a \textit{coupling}: a measure in the set of product measures on $M\times M$ having marginals $\pi$ and $\mu$, i.e. $\gamma(E \times M) = \pi(E)$ and $\gamma(M \times E) = \mu(E)$ for all measurable $E \subset M$. \cite{ambrosio_users_2013}.
In the present setting, one of the two distributions under consideration is always a delta measure, $\mu=\delta_{\alpha}$, which severely constrains $\Gamma[\pi,\mu]$ and leads to efficient algorithms for important special cases.

\subsection{Finite State Spaces}
\label{ssec:computing_finite_state_spaces}
    Consider a finite state space $M=\{0, \ldots, p-1\}$ and metric specified by the $p\times p$ symmetric matrix $\bm{D}[i,j]=d(i,j)$. Let the matrix $\bm{C}[i,j] = d(i,j)^2$ encode the squared distances between states. Then the squared Wasserstein-2 distance between a distribution $\pi$ and a delta distribution $\delta_i$ on $M$ is simply $(\bm{C}\pi)[i]$, the $i$-th coordinate of the vector $\bm{C}\pi$. Therefore, by definition, the (squared) normalization constant associated with $F(\pi)$ will be 
    \begin{equation}
    \nu_{(M,d)}^2 = \max_{\pi \in \mathcal{P}(M)} \left\{\min_{0 \leq i \leq p-1} \sum_{j=0}^{p-1}\bm{C}[i,j]\pi[j]\right\}.
    \label{eq:norm_const_finite_state}
    \end{equation}
    In this case, $\mathcal{P}(M)$ is the set of length-$p$ probability vectors (with non-negative components which sum to 1). 
    Importantly, the optimization problem in  Eq.~\eqref{eq:norm_const_finite_state} can be realized as a linear programming (LP) problem, for which a global optimum is guaranteed to be found and for which many highly optimized methods exist \cite{Huangfu2018}.
    \begin{restatable}{prop}{lpproblem}
        For a finite metric space $(M,d)$ with $p$ states, the distribution with maximal variance can be determined by a linear program (LP) in $p+1$ variables, with $1$ equality constraint and $2p$ inequality constraints.
    \label{prop:finite_state_lp_problem}
    \end{restatable}
     The precise LP problem referenced to in Proposition~\ref{prop:finite_state_lp_problem} is provided in Appendix \ref{app:finite_state_lp_problem}. We have implemented routines in Python to compute the normalization constant for any finite state space \cite{modulerepo} that make use of fast LP solvers provided in the the Scipy Python module \cite{2020SciPy-NMeth}. Empirical computational costs of our implementations for example spaces are provided in Section \ref{ssec:implementations}.
     
    \subsection{Empirical Distributions on the Circle}
    \label{ssec:computing_empirical_circle}
    Having established methods for computing synchrony for any finite state space using results from optimal transport and LP, we turn our attention to an important continuous state space: the circle. Often, in practice, a population distribution will be an empirical sample distribution consisting of finitely many point masses (e.g., flow cytometry data \cite{McKinnon2018}). Given $\alpha \in S^1 \cong [0,1)$ and an empirical measure on the circle,  $\pi = \frac{1}{n}\sum_{i=1}^{n} \delta_{x_i}$ (supported on a finite sample $\{x_i\}_{i=1}^{n} \subset [0,1)$, which may have repeated elements), the squared Wasserstein-2 distance between $\delta_\alpha$ and $\pi$ is
    \begin{align}
        W_2^2(\delta_{\alpha},\pi) & = \frac{1}{n} \sum_{i=1}^n d(x_i,\alpha)^2 \nonumber \\
        & = \frac{1}{n} \sum_{i=1}^n \min(|x_i-\alpha|, 1-|x_i-\alpha|)^2 \nonumber \\
        & = \frac{1}{n} \sum_{i=1}^n \begin{cases} |x_i-\alpha|^2, & \text{ if } |x_i-\alpha| \leq 1/2 \\ 
        (1-|x_i-\alpha|)^2, & \text{ if } |x_i-\alpha| > 1/2.
        \end{cases}
        \label{eq:wasserstein2_circle_formula}
    \end{align}    
    We have already determined the synchrony normalization constant to be $\sqrt{3}/6$ for the circle of circumference 1 thanks to Proposition \ref{prop:circle_var_max}. It remains to determine the choice of $\alpha\in S^1$ which minimizes the function $W_2^2(\delta_{\alpha}, \pi)$ for the given empirical distribution, $\pi$. 
    
    Ultimately, this task is a constrained optimization problem in one variable, $\alpha \in [0,1)$, for which numerous global optimization methods are well suited. However, such methods generally do not come with guarantees on finding a global optimum for non-convex functions. Although  $W_2^2(\delta_{\alpha}, \pi)$ will not be convex on $[0,1)$, it will always be convex on subintervals between consecutive antipodal points, $a(x_i) = (x_i+1/2) \mod 1$, that partition [0,1) (see Fig.~\ref{supp_fig:w22_circle_example} for an example). This observation leads to an algorithm that is guaranteed to find the generalized variance of $\pi$ in a finite number of evaluations of Eq.~\eqref{eq:wasserstein2_circle_formula} that grows linearly with the number of point masses. This is made precise in the following proposition. 
    \begin{restatable}{prop}{empiricalcirclealgo}
    Let $\pi = \frac{1}{n}\sum_{i=1}^{n} \delta_{x_i}$ be an empirical distribution on the circle of circumference 1, supported on $\{x_i\}_{i=1}^{n} \subset [0,1)$. Then $F(\pi)$ is computable in at most $2n+1$ evaluations of $W_2^2(\delta_{\alpha}, \pi)$.
    \label{prop:complexity_empirical_circle}
    \end{restatable}
    The algorithm to determine the synchrony of a finite empirical distribution on the circle, whose correctness is established in the proof of Proposition~\ref{prop:complexity_empirical_circle}  in Appendix~\ref{app:proof_complexity_empirical_circle}, relies on computing the locations of the exact minimizers in each subinterval between consecutive antipodal points. These are either the critical points at which $dW_2^2(\delta_{\alpha}, \pi)/d\alpha = 0$ or the endpoints of the intervals in the partition of [0,1) by antipodal points, giving at most $2n+1$ candidate barycenters of $\pi$.


\section{Computational Results}\label{sec:results}
\subsection{Implementations}
\label{ssec:implementations}
    We implemented routines for computing synchrony measures for distributions on general finite state spaces, including discrete circular spaces, and for empirical distributions on the circle. These are provided in an open-source Python module \cite{modulerepo}. To assess the practical efficiency of computing synchrony, we conducted several experiments to measure the time cost of these implementations for varying numbers of states in finite state spaces (Fig.~\ref{fig:implementation_timings}A), and separately for the size of the support of empirical distributions on the circle (Fig.~\ref{fig:num_evals_exact_vs_approx}A). All timing experiments were conducted on an Intel(R) Core(TM) i9-7920X CPU with 12 cores at 2.90GHz, (running at approximately 4GHz).

    \begin{figure}[!ht]
    \includegraphics[width=1.0\textwidth]{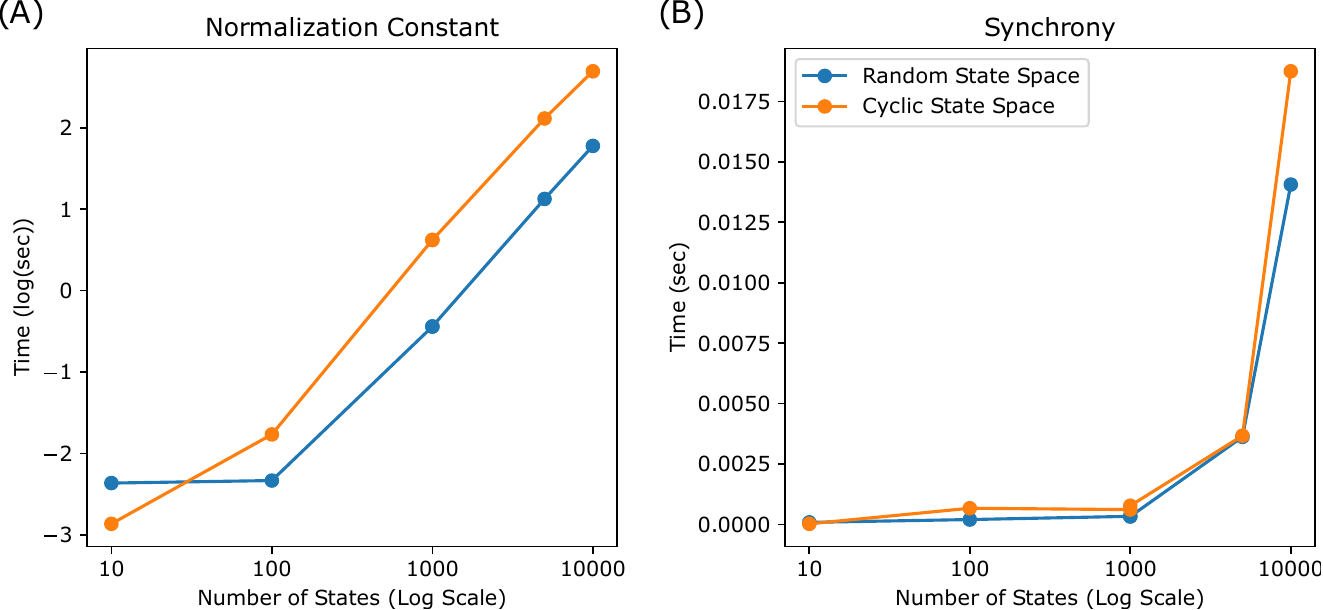}
    \caption{Times to compute the normalization constants (A) and the synchrony measures (B) for finite state spaces.}
    \label{fig:implementation_timings}
    \end{figure}

    In the case of a finite state space, the calculation of the normalization constant via solving an LP optimization problem is the most costly operation. That said, for a fixed state space, this constant need be computed only once. Moreover, LP problems are widespread and demand for solvers has led to numerous high-efficiency implementations. Our implementation leverages LP solvers provided in the open-source Python module, Scipy \cite{2020SciPy-NMeth}, which makes use of either high performance dual revised simplex or interior-point methods written in C++ \cite{Huangfu2018}. 
    
    We found that for even large state spaces the normalization constant can be computed in moderate time ($\sim$60 seconds for a random Euclidean state space with 10,000 states, Fig.~\ref{fig:implementation_timings}A). Interestingly, we found that the effective compute time of the normalization constant can depend significantly on the particular finite metric space. We observed approximately an order of magnitude increase in the average time to solve for the normalization constant of a uniform finite cyclic state space compared to random Euclidean state spaces of the same size, for large state spaces (Fig.~\ref{fig:implementation_timings}A). Here a random Euclidean state space was generated by sampling $n \in \{10, 100, 1000, 10000\}$ points uniformly at random from $[0,1)^{12}$ and treating each point as a state. The distance between two states was taken to be the standard Euclidean distance between the points. It should be noted that computing the normalization constant for a uniform cyclic state space is unnecessary given that the exact normalization is achieved by the uniform distribution and can be exactly computed without appealing to LP (Proposition \ref{prop:finite_cyclic_var_max}). 
    
    Having found the finite state space normalization constant, calculation of synchrony of a distribution is extremely fast (on the order of 10 milliseconds for a state space with 10000 states), as it requires only a single matrix-vector multiplication and finding the minimum entry in the resulting vector (Fig.~\ref{fig:implementation_timings}B).

    The theoretical results in Section \ref{sec:computing_asynch} guarantee that the exact value of synchrony of an empirical distribution supported on $n$ points in the circle can be determined by a number of evaluations of the Wasserstein-2 distance that grows linearly in $n$. However, we can sacrifice the guarantee of correctness to potentially gain computational efficiency by approximating synchrony using a global optimization algorithm. Our implementation of the routine which computes synchrony of an empirical distribution on the circle \cite{modulerepo} uses either the exact method established in Appendix~\ref{app:proof_complexity_empirical_circle}, or an approximate method using a basin-hopping global optimization algorithm, implemented in the Python module Scipy \cite{2020SciPy-NMeth} based on the method described in \cite{basinhopping}.
    
    By using a basin-hopping global optimization method \cite{basinhopping}, we found that for uniform random data, on average, less than 2500 evaluations of the Wasserstein-2 distance were required to find an estimate near double precision machine epsilon (max error $\approx 2.22e-16$, not shown) of the exact measure of synchrony for empirical distributions consisting of up to 10,000 point masses (Fig.~\ref{fig:num_evals_exact_vs_approx}B). This results in a significantly lower synchrony compute time than the exact calculation, for large numbers of point masses on the circle, and apparently sublinear growth with the number of point masses (Fig.~\ref{fig:num_evals_exact_vs_approx}A).
    
    \begin{figure}[!ht]
    \includegraphics[width=1.0\textwidth]{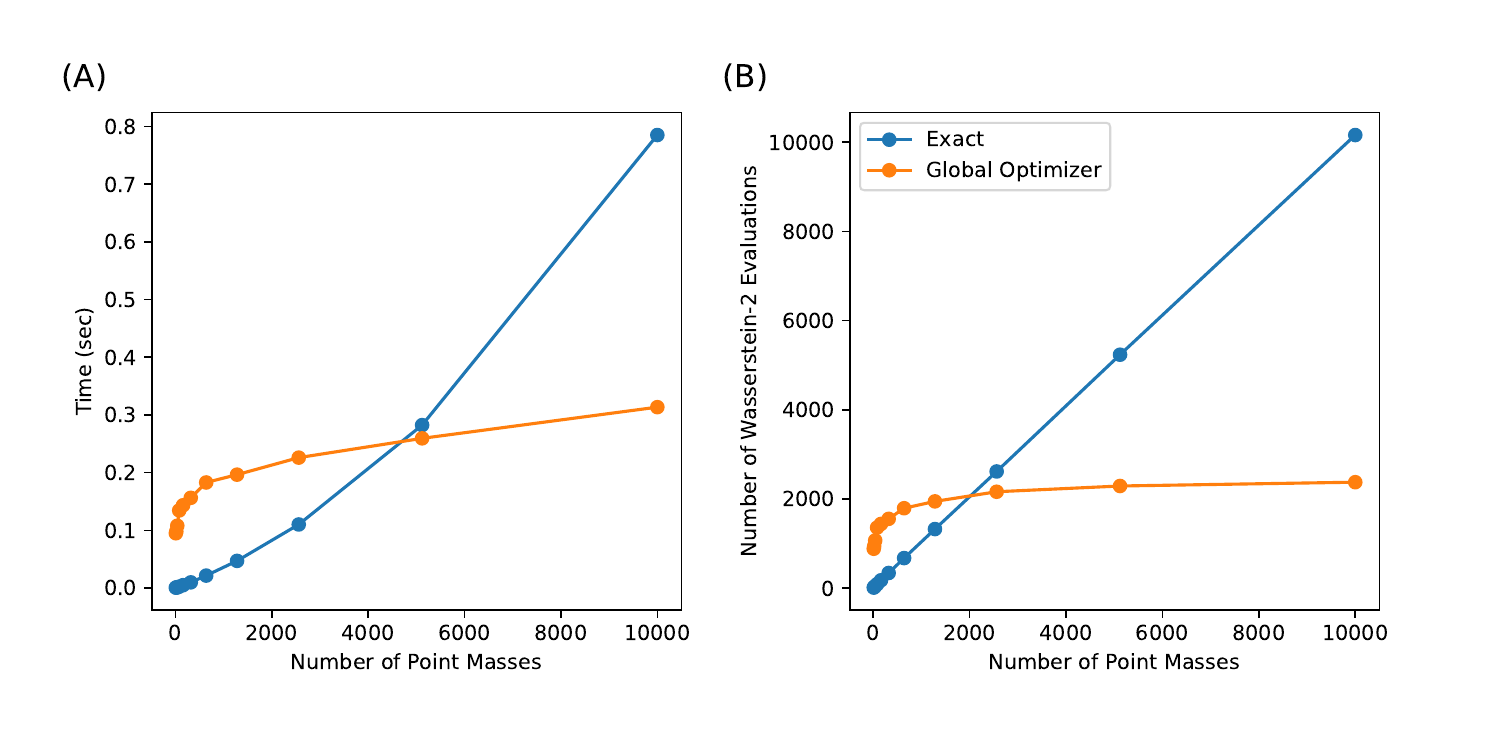}
    \caption{(A) Time to compute synchrony measure for empirical distributions on the circle with $n$ points using either the exact, provably correct algorithm or a basin-hopping global optimization algorithm. (B) The numbers of evaluations of the Wasserstein-2 metric between the empirical distribution and a delta distribution on the circle.}
    \label{fig:num_evals_exact_vs_approx}
    \end{figure}

\subsection{Dynamics of Synchrony}
\label{ssec:dynamics}

Many applications focus on populations of individuals \textit{with dynamics}, meaning that the state of an individual may change over time. The first and most basic observation about synchrony in this setting is that it too may vary over time as the distribution of the population across states changes due to variation in the dynamics of individuals.

\subsubsection{Loss of Synchrony Due to Variable Progression Rates}
\label{sec:synchloss_var}
Many common biological experiments derive measurements from populations of cells (e.g., RNA-Seq \cite{Wang2009}). Often the goal of such experiments is to make inferences about individuals within the population, and thus synchronization procedures are required to reduce the effect of the convolution of states (e.g., gene expression) occupied by a distribution of individuals in a cell population. For example, elutriation by centrifuge can be used to synchronize a population of yeast cells in their cellular division cycle \cite{Rosebrock2017SynchronizationOB}, or cells can be chemically arrested in a particular phase \cite{BREEDEN1997332}. Circadian-clock synchrony can be similarly chemically-induced in mammalian tissue cultures \cite{BALSALOBRE1998929}, and numerous methods exist to synchronize \textit{Plasmodium} parasites in their IEC \cite{Ranford-Cartwright2010}.  Each of these cases represents a periodic biological process, and in each case it is well known that population synchrony will degrade over time due to a variety of factors including asynchronous cellular development \cite{RN1716, RN1439}, and/or natural biological variation in intrinsic periods \cite{science, RN2239}. 

To illustrate and quantify the loss of population synchrony of some periodic processes due to variability in cycle progression rates, we measured synchrony over time of a simple model of a population of phase oscillators. Here the observed state space $S=S^1 \cong [0,1)$ represents cycle phase, and the $i$th individual parasite's state at time $t$, $x_i(t)$, changes at a constant rate, so 
$$
dx_i(t)/dt = \beta_i > 0.
$$
We expect the rate at which the population loses synchrony will increase with increasing variance in cycle progression rates across the population. Fig.~\ref{fig:phase_oscillator_synchloss}(A) shows the dynamics of synchrony loss of a population of 10,000 oscillators whose phase velocities are drawn from a log-logistic distribution,
$$
\beta_i \sim \text{LL}(x;\mu,\sigma), 
$$
 for three choices of scale parameter $\sigma$ in the density function
$$
\text{LL}(x;\mu,\sigma) = \frac{\exp(z)}{\sigma x(1+\exp(z))^2}, \; \; z:=(\ln(x)-\mu)/\sigma,
$$
where $\mu = \ln(\pi\sigma/\sin(\pi\sigma))$. The choices of $\mu$ and $\sigma$ enforce that the expected period length is always 1 (cycle), and the exact coefficient of variation (CoV, $\sqrt{\tan(\pi\sigma)/(\pi\sigma)-1}$) of the phase-velocity distribution is either 0.1, 0.2, or 0.5 (corresponding to $\sigma \approx$ .054805, 0.107709, and 0.241695 respectively). Each population was initially perfectly synchronized in phase 0. 

    \begin{figure}[!ht]
    \includegraphics[width=1.0\textwidth]{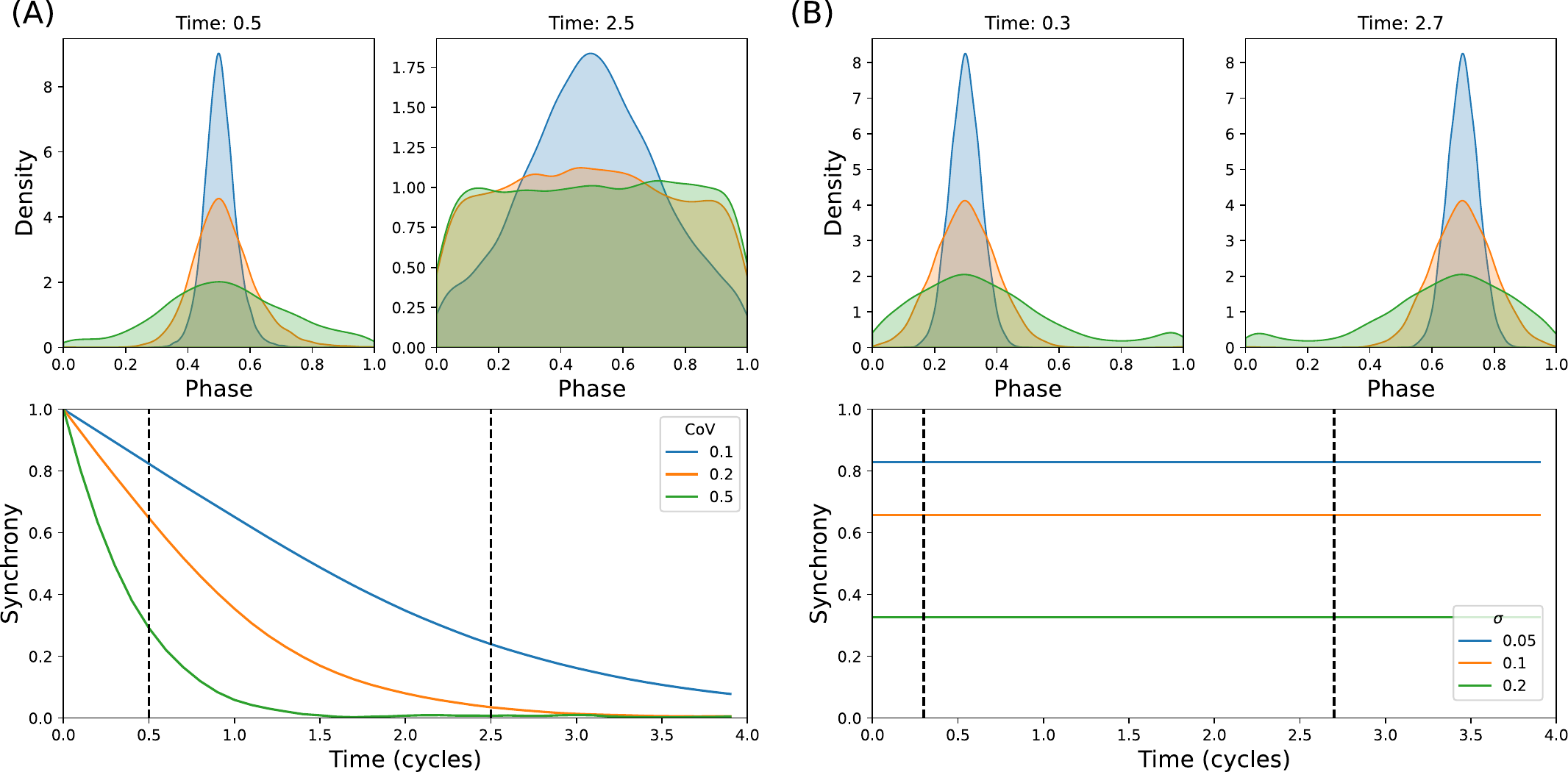}
    \caption{(A) Phase distributions of a population of 10,000 oscillators (top row) with individuals having different phase velocities at time t=0.5 (left) and t=2.5 (right), and their synchrony over time (bottom row). Initial phases were set to 0 for all individuals. Phase velocities were drawn from log-logistic distributions with coefficients of variation (CoV) either 0.1, 0.2, or 0.5. (B) Phase distributions of a population of 10,000 oscillators (top row) with constant (1) phase velocity at time t=0.3 (left) and t=2.7 (right), and their synchrony over time (bottom row). Initial phases were drawn from normal distributions with variances, $\sigma$ either 0.05, 0.1, or 0.2 wrapped to the interval [0,1).}
    \label{fig:phase_oscillator_synchloss}
    \end{figure}

Although all populations converge towards a steady-state uniform distribution, the rate of convergence increases with greater cycle period variability, as expected. For example, with a population of phase oscillators whose velocities exhibit a coefficient of variation of 10\% of the average period length, synchrony has dropped from 1 to approximately 0.5 after 1.5 cycles. On the other hand, given a larger variability in cycle progression rates (CoV 0.5), by 1.5 cycles the population is essentially uniformly distributed (Fig.~\ref{fig:phase_oscillator_synchloss}A). 

Of course, if $\beta_i = \beta$ is constant across the population, synchrony will not change in time: the barycenter of the initial population will simply progress at the same rate $\beta$, and in this rotating reference frame, the distribution will not change. In Fig.~\ref{fig:phase_oscillator_synchloss}(B), three populations of phase oscillators with different amounts of initial synchrony were produced by sampling initial phases from the mean-0 normal distribution with density 
$$
N(x;\sigma) = \exp(-x^2/(2\sigma))/(\sigma\sqrt{2\pi}),
$$
and wrapping to the unit interval by reducing samples modulo-1. Each individual progressed through its cycle phases at the same rate, resulting in constant synchrony over time. Note that a population's (constant) synchrony decreases with increasing initial phase variance, $\sigma$, chosen in Fig.~\ref{fig:phase_oscillator_synchloss}(B) to be either 0.05, 0.1, or 0.2. For $\sigma=0.05$, synchrony is over 0.8, indicative of the a highly concentrated population, while for $\sigma=0.2$ the distribution has synchrony less than 0.4, reflecting the large variance of the broad distribution over cycle phases.


\subsubsection{The Impact of Replication on Population Synchrony}
\label{sec:replication}

For cell division cycles and \textit{Plamodium} IECs, the population size and distribution across states are changing over time due to replication. In the latter case, replication occurs at a particular phase of development and schizogony results in anywhere between 2 to 32 daughter parasites \cite{RN2240}. It has been suggested that replication induces a biased assessment of population synchrony, using microscopy measurements that identify cellular state by morphological characteristics (so-called \textit{stage percentage curves}) \cite{GREISCHAR2019341}.  We clarify this observation through a reanalysis of the heuristic model put forward in \cite{GREISCHAR2019341}, using the rigorous definition of population synchrony presented here.

    \begin{figure}[!ht]
    \includegraphics[width=1.0\textwidth]{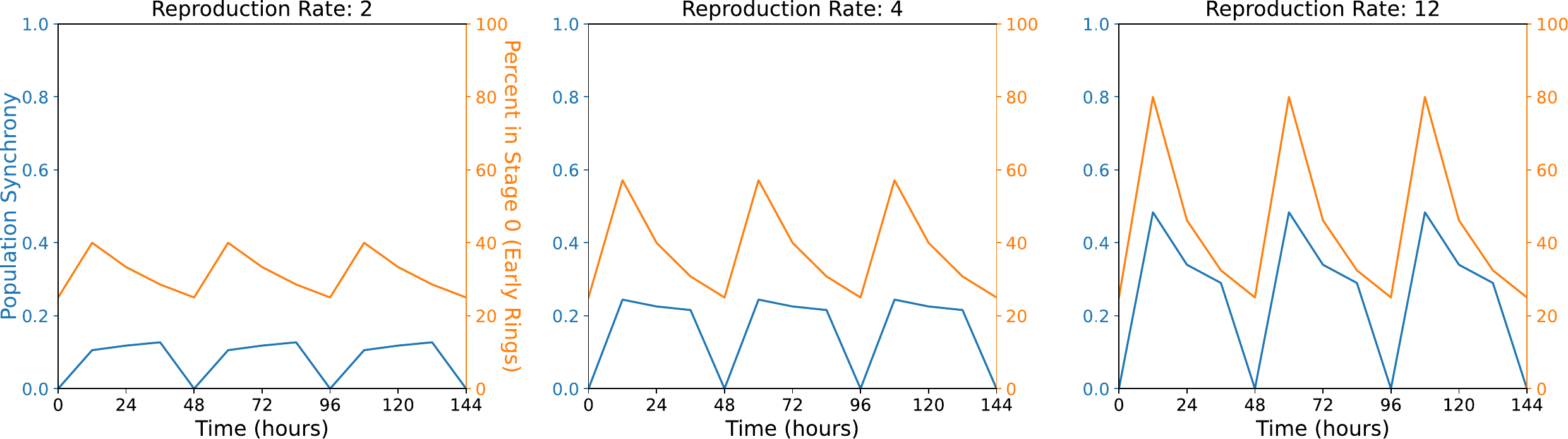}
    \caption{Synchrony over time of a heuristic discrete-time model of a replicating population of \textit{Plasmodium} parasites in a uniform discrete circular space consisting of 4 states (0=early ring, 1=late ring, 2=trophozoite, and 3=schizont). Parasites are assumed to progress at a constant rate from stage $i$ to stage $i+1 \; (\text{mod } 4)$ and replicate with the given reproduction rate upon leaving stage 3 (schizont) and entering stage 0 (early ring). Also shown is the stage 0 percentage curve for three different choices of reproduction rates ($s=2,4,$ and $12$).}
    \label{fig:greischar_model_popsynch}
    \end{figure}

Consider a population of malaria parasites progressing through the IEC, encoded as a time-dependent length-4 vector $\bm{x}$, whose integral components count the number of parasites in each of four distinct stages (early ring, late ring, trophozoite, and schizont). If we assume no variability in the progression rates of the parasites, this population can be modeled by 
\begin{equation}
\bm{x}(t+1) = \bm{x}(t)
\begin{bmatrix}
0  &  1  &  0  &  0  \\ 
0  &  0  &  1  &  0  \\
0  &  0  &  0  &  1  \\
s  &  0  &  0  &  0  \\
\end{bmatrix},
\label{eq:greischar_model}
\end{equation}
where $s \geq 1$ is the reproduction rate at the schizont-to-early ring transition \cite{GREISCHAR2019341}. In this setting the population distribution at time $t$ is then
\begin{equation*}
\pi(t) = \bm{x}(t)/N(t),
\end{equation*}
where $N(t)$ is the total size of the population at time $t$, so that the coordinates of $\pi(t)$ measure the fraction of parasites in each stage, which is the measurement often made in this context. 

When $s=1$, the population does not experience replication. Instead $\bm{x}(t+1)[i] = \bm{x}(t)[i-1]$ (where indices are taken modulo 4). The result is a cyclic shift of the masses in each state at each time step, which yields a constant degree of synchrony over time. However, if $s > 1$, even an initially uniformly distributed (maximally asynchronous) population (i.e., $\pi(0) = [1/4,1/4,1/4,1/4]$) will show periodic fluctuations in the stage percentage curves due to replication, and the extent of the fluctuation will depend on the size of the reproduction rate. This phenomenon is naturally reflected by our measure of population synchrony. Fig.~\ref{fig:greischar_model_popsynch} reports the population synchrony in the uniform discrete circular state space (defined in Example \ref{exp:discrete_circle_uniform}) of a population of parasites initially uniformly distributed and moving at a constant rate across 4 morphological stages, according to Eq.~\eqref{eq:greischar_model}. Fig.~\ref{fig:greischar_model_popsynch} also shows the percentage of parasites in stage 0 (early ring), which too fluctuates periodically over time. The reason that the population is becoming periodically more and less synchronous in time is because the population's distribution over morphological stages changes periodically over time.  We note, as was observed in \cite{GREISCHAR2019341}, that the maximum percentage of parasites in the early ring stage over one period grows with the reproduction rate, $s$. This is expected since the population can become arbitrarily highly concentrated in the first stage following replication with a sufficiently large choice of reproduction rate. 

Thus replication need not result in an overestimate of synchrony if defined as a measure of how well a population is aligned in some state space at each moment in time. Instead, replication induces dynamic changes in population synchrony--even in a population with zero variability in its phase velocities.


\subsubsection{Boundary Effects Due to Discrete Observations}
\label{ssec:boundary}

The examples in Sections \ref{sec:synchloss_var} 
 and \ref{sec:replication} remind us that synchrony, if measured at only one time point, may be insufficient to characterize synchrony as it evolves over time. In this section, we illustrate how measurement limitations can also affect the accuracy of estimating population synchrony, especially when limiting measurements to a single time point.

In the case of \textit{Plasmodium} parasites, the developmental phase is usually approximated by an assignment of morphological stage. We regard this as a discretization of the continuous circle (whose states are phases measuring the fraction of the cycle length traversed) into disjoint arcs over which observed morphology is invariant (see Example~\ref{exp:discrete_circle_general}). Unlike replication, the loss of information caused by discretization can introduce substantial bias into the observed progression of synchrony, with synchrony periodically over- and under-estimated. Specifically, spurious oscillations in synchrony can be observed as population distributions traverse sharp stage boundaries (Fig.~\ref{fig:synchrony_loss_discretization}).


    \begin{figure}[!ht]

\begin{tabular}{c}
\begin{tikzpicture}
    \begin{scope}[shift={(0,0)}]
        \draw[thick] (0.75,0) rectangle (1.2,2);
        
        \draw[thick] (0,0) -- (6,0);
        
        \foreach \x in {1.5, 3, 4.5}
            \draw[dashed] (\x,0) -- (\x,2.5);
            
        \node[font=\tiny] at (0.5,2.5) {State 0};
        \node[font=\tiny] at (2.25,2.5) {State 1};
        \node[font=\tiny] at (3.75,2.5) {State 2};
        \node[font=\tiny] at (5.5,2.5) {State 3};
    \end{scope}
    
    \begin{scope}[shift={(7,0)}]
        \draw[thick] (1.25,0) rectangle (1.7,2);
        
        \draw[thick] (0,0) -- (6,0);
        
        \draw[dashed] (1.5,0) -- (1.5,2.5);
        \foreach \x in {3, 4.5}
            \draw[dashed] (\x,0) -- (\x,2.5);
            
        \node[font=\tiny] at (0.5,2.5) {State 0};
        \node[font=\tiny] at (2.25,2.5) {State 1};
        \node[font=\tiny] at (3.75,2.5) {State 2};
        \node[font=\tiny] at (5.5,2.5) {State 3};
    \end{scope}
\end{tikzpicture}

\\

 (A)  \hspace{2.5in} (B)

\\

\begin{tikzpicture}
    \draw[thick] (0,0) -- (6,0);
    
    \foreach \x in {1.5, 3, 4.5} {
        \draw[dashed] (\x,0) -- (\x,2.5); 
        
        \draw[thick] (\x-0.25,1.5) -- (\x,0.4) -- (\x+0.25,1.5);
    }
    
    \foreach \x in {0.25, 1.75, 3.25, 4.75} {
        \draw[thick] (\x,1.5) -- (\x+1,1.5);
    }

    \node at (2.5,-0.5) {Time};
    \draw[thick, -stealth] (3.25,-0.5) -- (4.5,-0.5); 

    \node[rotate=90] at (-0.5,1) {$F(\pi)$};

\end{tikzpicture}

\\

(C)

\end{tabular}
\caption{Periodic synchrony due to the discretization of $S^1$. Panels (A) and (B) represent two observations of the same population separated by some period of time. The shape of the distribution of individuals $\pi(t)$ remains unchanged, so in reality, the population has constant synchrony. However, if the underlying continuous space is discretized into intervals, then perfect synchrony (all individuals at state 0) in panel (A) decays in panel (B) (half the individuals in state 0 and the other half in state 1). Panel (C): This apparent change from perfect synchrony to substantial asynchrony repeats periodically over time. Here, the dotted lines represent moments where $\pi(t)$ is equidistant across a boundary.}\label{fig:synchrony_loss_discretization}

    \end{figure}
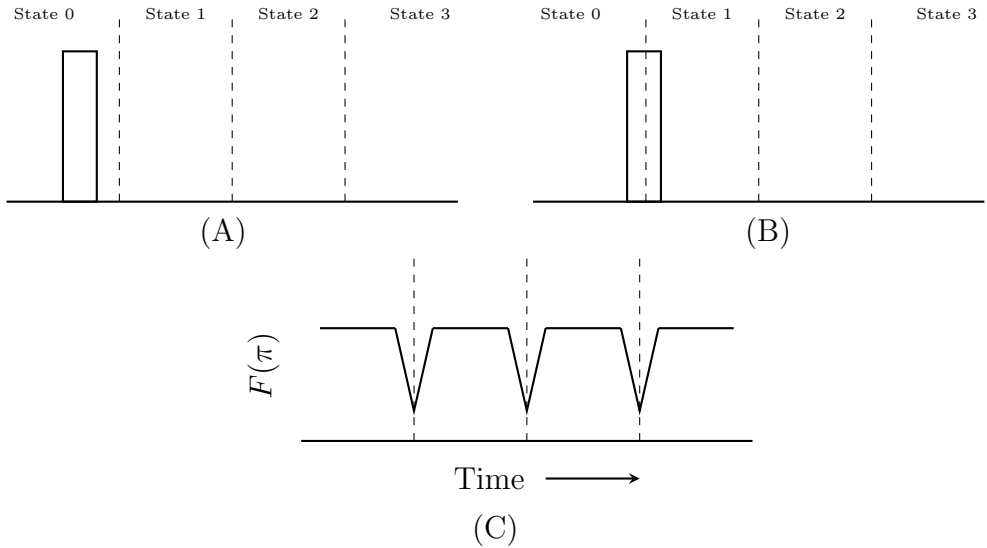
    
Suppose a population is distributed according to $\pi(t)$ as a function of time in the true state space $M$, where $M$ is the circle parameterized by $[0,1)$, and $\pi(t)$ is the uniform distribution over the interval $[t,t+1/8]$ (taken modulo $1$). Under this assumption, synchrony  is constant in time, $F(\pi(t)) = 95/96$ for all $t$.

However, suppose that there are only 4 observable states equally dividing the circle, so that events are classified according to the state space $X = \{0,1,2,3\}$ as in Section~\ref{sec:replication}. If state $i$ corresponds to the arc $[i/4,(i+1)/4), i=0,\ldots, 3$, the observed population distribution $\pi_X(t)$ will then consist of either a single state $i$, whenever $[t,t+1/8] \subset [i/4,(i+1)/4)$, leading to $F(\pi')=1$, or a population split between two states whenever $i/4 \in (t,t+1/8)$ modulo 1. By symmetry, the smallest observed synchrony will occur when the population is evenly split across the boundary of two discrete states, resulting in an observed synchrony of $1-1/\sqrt{3}$. The synchrony measure $F(\pi'(t))$ is given schematically in Fig.~\ref{fig:synchrony_loss_discretization}, showing the linear excursions from perfect synchrony as the true population traverses discrete state boundaries. 
In other words, the observed distribution appears to be (periodically) perfectly synchronized in $S$, even though the unobserved true distribution is neither perfectly synchronized nor dynamic in $M$.



To further illustrate the effect of discretization of a state space, we performed a reanalysis of a phase oscillator model---originally reported in \cite{science} and described in Section \ref{sec:synchloss_var}---and experimentally-measured time series stage percentage curves for several strains of \textit{P. falciparum}. We computed synchrony for each model phase oscillator population, thought of as time-evolving empirical distributions on the circle. We also computed synchrony for the experimentally-measured and the model-predicted stage percentage curves using finite circular state spaces, with 4 states determined by the strain-specific stage threshold phases reported in \cite{science}. For completeness these thresholds are provided in Table~\ref{supp_tab:pfalc_stage_thresholds} in Appendix \ref{supp:pfalc_stage_spaces}. The distances between states were taken to be the arc lengths between the midpoints of the model-inferred stage intervals as defined in Example \ref{exp:discrete_circle_general}.

The cartoon example in Fig.~\ref{fig:synchrony_loss_discretization}, suggests that synchrony measured in the observable finite state space exhibits qualitatively different dynamics than synchrony measured in the true underlying state space. In Fig.~\ref{fig:science_reanalysis} we observe synchrony smoothly degrading in time in the phase oscillator model (due to non-zero variance in phase velocities across the population, as in Fig.~\ref{fig:phase_oscillator_synchloss}(A)). However, the discretization of the circle into morphological stages induces rapid decreases and increases in synchrony as the population's distribution in phase space crosses stage thresholds. Thus, a single measurement of population synchrony in the coarse 3-stage discretization of the continuous phase space dramatically misrepresents the degree of synchrony of the population in the latter. In Fig.~\ref{fig:science_reanalysis} we observed examples where the former state space reports a much lower degree of synchrony than the latter at many time points (e.g., HB3 at 45 hours), and marginally higher synchrony at others (e.g., FVO at 60 hours). 

    \begin{figure}[!ht]
    \includegraphics[width=1.0\textwidth]{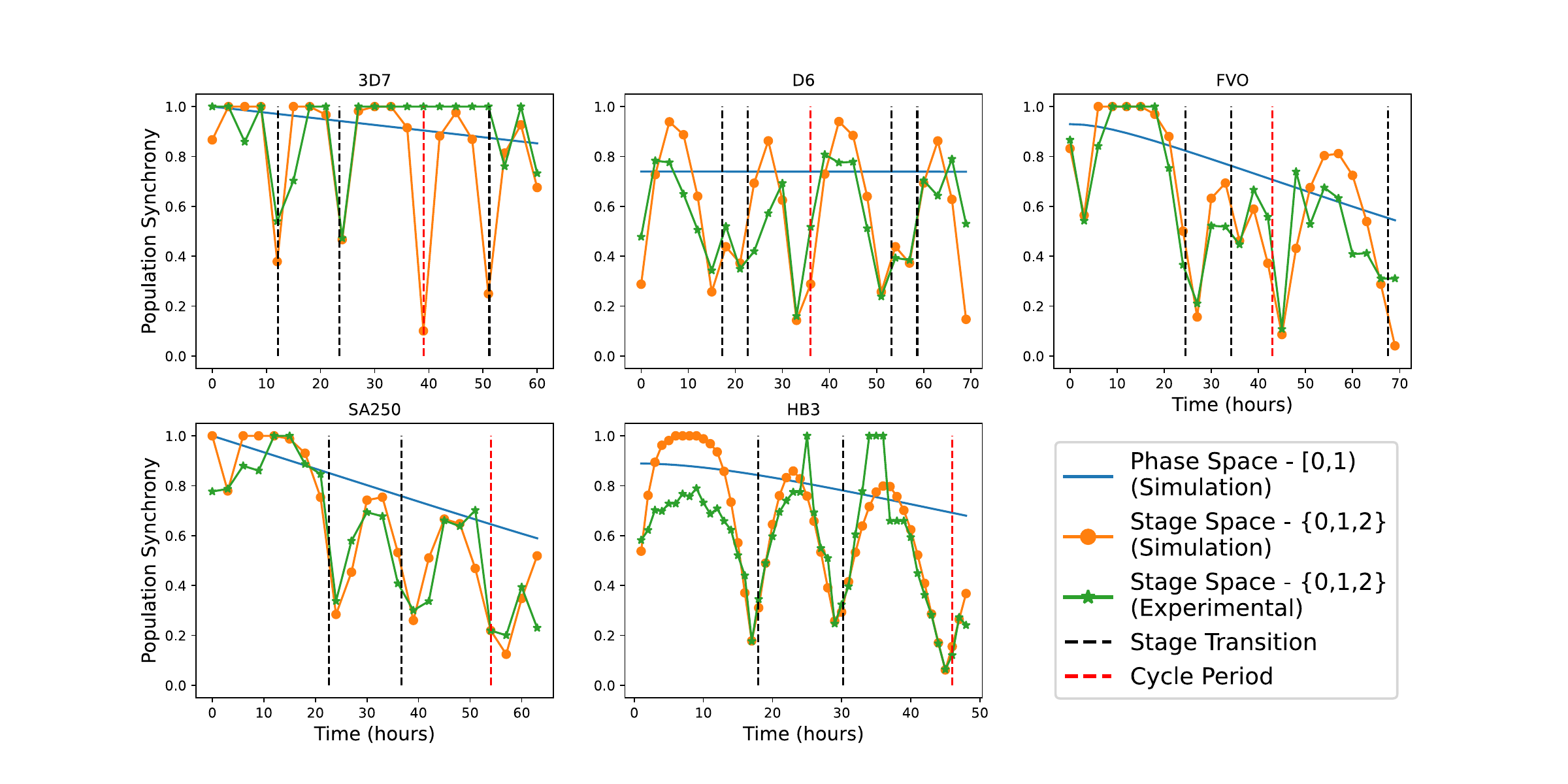}
    \caption{Measures of synchrony over time of simulated and experimentally-observed populations of \textit{P. falciparum} parasites, reported in \cite{science}. Modeled parasite populations consist of 10000 phase oscillators that progress through phases of the circle, [0,1), at randomly selected phase velocities. The circle is also discretized into morphological stages by strain-specific stage-transition thresholds, indicated by black dashed lines. Strain-specific cycle periods are indicated by red dashed lines. Strain-specific stage transition thresholds are provided in Appendix \ref{supp:pfalc_stage_spaces}.}
    \label{fig:science_reanalysis}
    \end{figure}

Despite the spurious effects of discretization of the state space in the data reported in Fig.~\ref{fig:science_reanalysis}, the maximum observed synchronies in the discrete circular spaces qualitatively track the degradation of synchrony in phase space over time caused by variable parasite IEC progression rates. 

Any spurious dynamics induced by the discrete nature of measurements may be unavoidable (and undetectable) in practice.  That said, by refining the discretization of a true state space, the effect may be greatly ameliorated. To demonstrate this, we refined the three \textit{Plasmodium} stages modeled in Fig.~\ref{fig:science_reanalysis}, by subdividing each of the stages' corresponding phase-intervals into 2 or 3 equal length bins. This yielded state spaces consisting of either 6 or 9 total stages, which could be regarded as further categorizing each morphological stage into early and late, or early, middle and late (sub)stages. The inter-state distance matrices for these state spaces are provided in Appendix  \ref{supp:pfalc_stage_spaces}. Synchrony was then computed for these (larger) discrete state spaces over time, as shown in Fig.~\ref{fig:science_reanalysis2}. For each strain, there is a reduction in the maximum deviations between synchrony measured in the continuous phase space and the 6-stage space when compared with the deviation between the phase-synchrony and synchrony measured in the coarser 3-stage space. This deviation is further reduced for the refined discretizations of the circle into 9 stages. 

These observations reflect the content of a special case of Proposition~\ref{prop:approximability}. In this setting, $M = S^1 \cong [0,1)$, which is divided into disjoint arcs, $[a_{i}, a_{i+1})$, representing observable stages (as in Example~\ref{exp:discrete_circle_general}) whose distances are determined by the arclength between midpoints of arcs. Equivalently, all phases in the arc $[a_{i}, a_{i+1})$ are assigned to the midpoint $x_i$, by $g:M\to X$, where  $X = \{x_{1},\ldots, x_{k}\}$ (for $k=3, 6$ or $9$) is a discretization of $M$. As $k$ increases, the maximum distance from any true state, $x$, to its observed state, $g(x)$, is reduced, leading to a generally better approximation of synchrony in the more refined observed spaces, as suggested by Proposition~\ref{prop:approximability}.

    \begin{figure}[!ht]
    \includegraphics[width=1.0\textwidth]{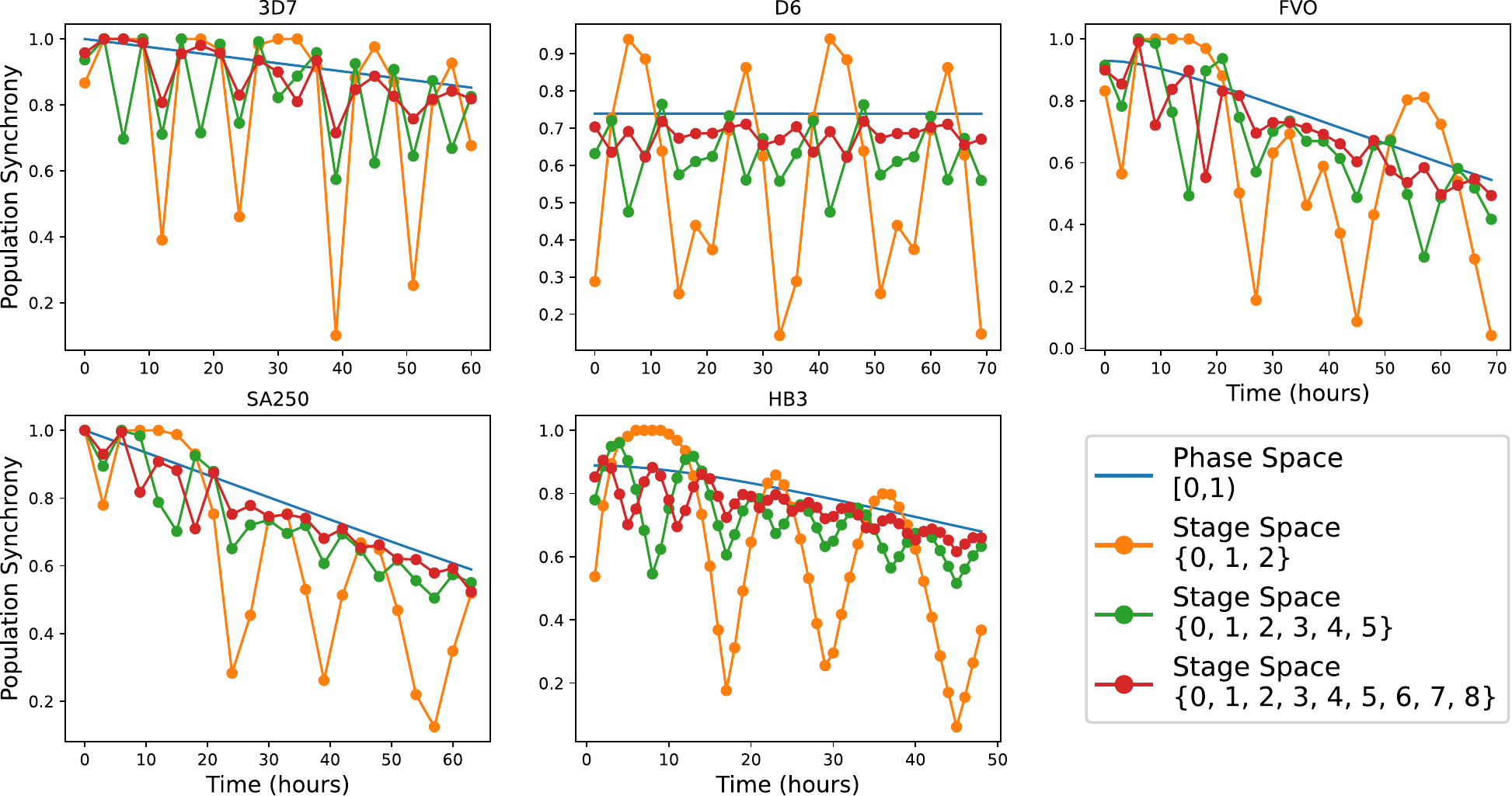}
    \caption{Measures of synchrony over time of simulated populations of \textit{P. falciparum} parasites. Modeled parasite populations consist of 10000 phase oscillators that progress through phases of the circle, [0,1), at randomly selected phase velocities. The circle is discreteized into 3, 6, or 9 morphological stages by strain-specific stage transition thresholds. Strain-specific stage transition thresholds or inter-stage distances for these discrete circular spaces are provided in Appendix \ref{supp:pfalc_stage_spaces}.}
    \label{fig:science_reanalysis2}
    \end{figure}


\section{Discussion}
\label{sec:discuss}

This paper introduces a general quantification of synchrony that captures the degree to which a population is synchronized in a chosen state space, which is required to be a compact metric space. This restriction ensures synchrony is well-defined and simplifies some of the mathematical analysis, while ensuring the results are general enough to apply broadly to many experimental settings and measurement types.

This work establishes numerous desirable mathematical properties of synchrony, including its insensitivity to measurement noise, invariance to scaling of the underlying state space, and its approximability when the true state space is not entirely accessible but is rather represented by only a subset of observable states. 

This final property highlights the synchrony measure's dependence on a choice of state space, which may be dictated by what is measurable in a given experimental context. We have shown through examples how discretization of a continuous state space can introduce error in the measure of synchrony; however, we have also shown that this error may be bounded by how well the observable states approximate the underlying space (Proposition \ref{prop:approximability}). We showcased this process through a reanalysis of a model of the discrete \textit{P. falciparum} stages measured in \cite{science}. Although such a hypothetical refinement of morphological stages is not typical in studies involving \textit{Plasmodium}, trained microscopists are able to distinguish between substages of early trophozoite stage in \textit{P. vivax}, as was done in \cite{Motta2023-pnas}. 

Population \textit{asynchrony} is related to the (normalized) minimum distance between a population distribution and any perfectly synchronized population in a state space, where distance is measured by the Wasserstein-2 distance between distributions from Optimal Transport \cite{KantorovichW2metric}. Normalization requires finding a distribution which is furthest from any delta distribution. This problem of maximizing the Wasserstein-$p$ distance to a given distribution, and the techniques and applications presented here, may be of interest to Optimal Transport researchers.

Proposition \ref{prop:approximability} is related to a large body of mathematical work concerning \textit{quantization of measure}, in which one attempts to solve for 
$$
V_{N,p}(\mu) := \inf_{\# \text{supp } \mu_N \leq N} W_{p}(\mu,\mu_N)^p
$$
and find the measure supported on at most $N$ points that ``best'' approximates a given measure, $\mu$, \cite{foundquantmeasure} (where ``best'' refers to the minimal Wasserstein-$p$ distance between the measures, $1 \leq p < \infty$). Foundational work establishes, under mild assumptions about the space supporting $\mu$, that $V_{N,p}(\mu) \rightarrow 0$ as $N$ tends to infinity (see \cite[Lemma 6.1, p. 77]{foundquantmeasure} for example). Much work has focused on determining the asymptotic rates of convergence of $V_{N,p}(\mu)$ in different settings, \cite{foundquantmeasure,Gruber2001,GRUBER2004456,COCV_2012__18_2_343_0, COCV_2016__22_3_770_0}. 
In contrast, Proposition~\ref{prop:approximability} does not report a rate of convergence of optimal quantized measures, but instead can be used to bound the error of particular quantized measures whose support is specified by a measurable function and whose preimages partition the space. In experimental settings we don't expect one will always have a choice of the support of the discretization. Indeed, in the examples presented in Section~\ref{sec:results},  \textit{Plasmodium} stages are determined by specific, discernible morphology characteristics, and so defining arbitrary stage thresholds by microscopy measurements may not be possible. 

 A recent study proposes quantifying synchrony of a population of \textit{Plasmodium} parasites during their IEC as 1 minus the ratio of the time required for 99\% of the parasite population to traverse the IEC to the IEC period length \cite{Greischar2024}. It is important to note that this measure is, in principle, influenced both by the initial distribution of parasites in the IEC and by the variability in how individual parasites progress through their IEC. That said, the analysis presented was limited to the simplified case in which all parasites are assumed to progress at the same rate. Thus, in \cite{Greischar2024}, the proposed measure captures only the population's initial synchrony. For this reason, and because the study was limited to synthetic data, the utility of the proposed synchrony measure to experimental measurements is unclear. 
 
 It has also been suggested that stage percentage curves present a biased measure of synchrony \cite{Greischar2024,GREISCHAR2019341} due to the effects of replication. This conclusion is based on the observation that even if a population of parasites is uniformly distributed in its state space and every individual progresses at the same rate through states, the proportion of parasites in ring stage will periodically increase and decrease \cite{GREISCHAR2019341}. We observe the same phenomenon in our proposed measure of synchrony: replication can induce dynamic changes in synchrony (Fig.~\ref{fig:greischar_model_popsynch}). We do not, however, conclude that this is a bias in synchrony itself. Instead, we regard it as an inevitable property of any reasonable measure of synchrony that depends on the population distribution, since, periodically, a greater percentage of individuals occupy early ring.
 
 We also expect that the periodic increase and decrease in population synchrony in response to replication depends critically on the (unrealistic) assumption that all parasites progress through states at the same rate. We expect the synchrony of a replicating population in which individuals progress at different rates will exhibit qualitatively different dynamics, and may approach a new steady state, as suggested by the non-replicating simulated populations in Fig~\ref{fig:phase_oscillator_synchloss}. Precise characterization of the impact of the modeling assumption that all parasites progress at the same rate is ongoing and beyond the scope of this paper.
 
 Ostensibly, replication-induced synchronization and desynchronization could influence our understanding of the variability with which individuals in a population progress through a state space. As a thought experiment, imagine two populations in a periodic state space, having identical initial distributions and identical variability in how individuals move through state space, but only one of which undergoes replication at a particular phase. Will replication cause the one population to lose synchrony more slowly over time than the other? Based on our findings, we expect the degree and nature of the influence of replication on the apparent changes in synchrony over time will depend on population-level factors, including 1) how the population is distributed, 2) how variably individuals progress through states; experimental factors, such as 3) when and how often measurements are taken, and 4) how the observable measurement relates to the underlying state space; and active mechanistic factors, such as 5) entrainment to a external signal(s) and/or 6) interactions between indivduals in the population. 
 
 We have shown that variability of individual progression rates can influence the dynamics of population synchrony (Fig.~\ref{fig:phase_oscillator_synchloss}), as can replication (Fig.~\ref{fig:greischar_model_popsynch}). We have also demonstrated the impact discretization of a state space can have on the dynamics of measured synchrony (Fig.~\ref{fig:synchrony_loss_discretization} and Fig.~\ref{fig:science_reanalysis}) and shown the influence of the fidelity of the discretization (Fig.~\ref{fig:science_reanalysis2}). Fully clarifying the interactions between variability of dynamics, replication, measurement limitations, and other mechanisms which may alter population synchrony is beyond the scope of this paper, but we believe that such investigations will be aided by the general, mathematically sound and rigorous definition of synchrony put forward here. Thus this work will contribute to the  growing body of work on modeling synchrony in biological populations \cite{science,RN1439, Guo2013-xh, Haase2023} and advance our understanding of the mechanisms of synchrony. 

Of particular interest, and the focus of future work, are the roles that coupling between individuals in a population and/or entrainment to an external signal may play in the establishment or maintenance of population synchrony. For example, circadian rhythms are expressed at the cellular level by intrinsic biological clocks which maintain only approximately 24-hour periods as individual oscillators \cite{RN1569}, and high levels of synchrony across diverse cell types in multicellular organisms is maintained by entrainment to external cues (e.g., light-dark cycles) \cite{RN1174}. Likewise, recent work provides evidence that \textit{Plasmodium} parasites have intrinsic clocks controlling precise periods of their IEC \cite{science} and that there is alignment between this clock phase and the circadian phase in humans \cite{Motta2023-pnas}. This observation suggests a conceptual model which may explain synchronous bursting of parasites in the IEC due to entrainment of parasite clocks to host circadian signals \cite{Motta2023-pnas}. We expect that the quantification of synchrony in this work will aid future modeling and experimental efforts that seek to elucidate conceptual and biochemical mechanisms responsible for maintaining synchrony in these and other related systems.

The notion of synchrony proposed in this work can be considered as a statistic of a probability distribution on a metric space. Given the large volume of work devoted to understanding statistics on manifolds and other metric spaces (e.g., see the books \cite{dryden2016statistical,pennec2019riemannian,small2012statistical}), we believe that a full investigation of the statistical properties of synchrony may provide another interesting avenue for future mathematics research.

\section*{Author contributions}

We adopt the CRediT taxonomy (https://credit.niso.org/) to specify author contributions: 

\noindent \textbf{Francis C. Motta:} Conceptualization, Methodology, Software, Validation, Formal analysis, Resources, Data Curation, Writing - Original Draft, Writing - Review \& Editing, Visualization, Project administration. \textbf{Kevin McGoff:} Conceptualization, Methodology,  Validation, Formal analysis, Writing - Original Draft, Writing - Review \& Editing \textbf{Breschine Cummins:} Conceptualization, Writing - Original Draft, Writing - Review \& Editing, Visualization \textbf{Steven B. Haase} Conceptualization, Writing - Review \& Editing.

\section*{Acknowledgements}

\noindent \textbf{Funding:} KM gratefully acknowledges funding from the National Science Foundation [DMS-1847144 and DMS-2113676].

\vspace{5mm}

\noindent \textbf{Declarations of interest:} Kevin McGoff reports financial support was provided by National Science Foundation. Steven B. Haase reports a relationship with Geometric Data Analytics, Inc. that includes: consulting or advisory and non-financial support. Francis Motta reports a relationship with Geometric Data Analytics, Inc. that includes: employment. Kevin McGoff reports a relationship with Geometric Data Analytics, Inc. that includes: employment. If there are other authors, they declare that they have no known competing financial interests or personal relationships that could have appeared to influence the work reported in this paper.

\section*{Code and Data Availability}
All data and code used to produce the results and figures in this study are provided in a public GitLab repository \cite{datacoderepo}. 






\appendix

\newpage
\section{Proof of Proposition \ref{prop:continuity}}
\label{app:proof_continuity}
\renewcommand{\thefigure}{A.\arabic{figure}}
\renewcommand{\theHfigure}{A.\arabic{figure}}
\renewcommand{\thetable}{A.\arabic{table}}
\renewcommand{\theHtable}{A.\arabic{table}}
\setcounter{figure}{0}
\setcounter{table}{0}

\continuity*

\begin{proof}
Because $(M,d)$ is compact, it follows that $\mathcal{P}(M)$ is also compact and metrizable with the weak$^*$ topology. By definition of $F$ (Eq.~\eqref{eq:synchrony}), it suffices to show that $V : \mathcal{P}(M) \to \R$ is continuous, where $V(\pi) = \Var(\pi)$.

For each $\alpha \in M$, define $f_{\alpha} : M \to \R$ by setting $f_{\alpha}(x) = d(x,\alpha)^2$. Note that $f_{\alpha}$ is bounded and continuous. Then define $F_{\alpha} : \mathcal{P}(M) \to \R$ by setting
\begin{equation*}
    F_{\alpha}(\pi) = \int f_{\alpha} \, d\pi.
\end{equation*}
Note that $F_{\alpha}$ is continuous (by definition of the weak$^*$ topology). Then we have
\begin{equation*}
    V(\pi) = \inf_{\alpha} F_{\alpha}(\pi),
\end{equation*}
and therefore $V$ is upper semi-continuous (as it is the pointwise infimum of continuous functions).

Let us now prove that $V$ is continuous.
Let $\{\pi_n\}$ be a sequence in $\mathcal{P}(M)$ that converges to $\pi$. Since $V$ is non-negative and upper semi-continuous, we have that 
\begin{equation*}
    0 \leq \limsup V(\pi_n) \leq V(\pi).
\end{equation*}
Let $L = \limsup V(\pi_n)$. We claim that
\begin{equation*}
    \lim V(\pi_n) = L = V(\pi).
\end{equation*}
In order to establish this fact, let us show that every subsequence $\{\pi_{n_k}\}$ has a further subsequence $\{\pi_{n_{k_j}}\}$ such that $\lim V(\pi_{n_{k_j}}) = L$ and $L = V(\pi)$.

Let $\{\pi_{n_k}\}$ be a subsequence of $\{\pi_n\}$. By the compactness of $M$ and the continuity of the map $\alpha \mapsto F_{\alpha}(\pi_{n_k})$ for each $k$ (resulting from the uniform dependence of $f_{\alpha}$ on $\alpha$), there exists $\alpha_k \in M$ such that 
\begin{equation*}
    V(\pi_{n_k}) = \int f_{\alpha_k} \, d\pi_{n_k}.
\end{equation*}
By the sequential compactness of $M$, there exists a subsequence $\{\alpha_{k_j}\}$ and a limit $\alpha \in M$ such that $\lim \alpha_{k_j} = \alpha$. Note that $\{f_{\alpha_{k_j}}\}$ converges uniformly on $M$ to $f_{\alpha}$.

Let $\epsilon >0$. Since $\{\pi_{n_{k_j}}\}$ converges to $\pi$ and $f_{\alpha}$ is continuous, there exists $J_1$ such that for all $j \geq J_1$, we have
\begin{equation*}
    \left| \int f_{\alpha} \, d \pi_{n_{k_j}} - \int f_{\alpha} \, d\pi \right| < \epsilon/2.
\end{equation*}
Also, since $\{f_{\alpha_{k_j}}\}$ converges uniformly to $f_{\alpha}$, there exists $J_2$ such that for all $j \geq J_2$ we have 
\begin{equation*}
\|f_{\alpha_{k_j}} - f_{\alpha}\| < \epsilon/2.
\end{equation*}
Then for $j \geq \max(J_1,J_2)$, we have
\begin{align*}
    \left| V(\pi_{n_{k_j}}) - F_{\alpha}(\pi) \right| & \leq \left| \int f_{\alpha_{k_j}} \, d\pi_{n_{k_j}} - \int f_{\alpha} \, d\pi_{n_{k_j}}\right| + \left| \int f_{\alpha} \, d\pi_{n_{k_j}} - \int f_{\alpha} \, d\pi \right| \\
    & \leq \| f_{\alpha_{k_j}} - f_{\alpha} \| + \epsilon/2 \\
    & < \epsilon/2 + \epsilon/2 = \epsilon.
\end{align*}
Thus, we have shown that $\lim V(\pi_{n_{k_j}}) = F_{\alpha}(\pi)$.
Then by this fact, the definition of $L$, the upper semi-continuity of $V$, and the definition of $V$ as an infimum, we see that
\begin{equation*}
    F_{\alpha}(\pi) = \lim V(\pi_{n_{k_j}}) \leq L \leq V(\pi) \leq F_{\alpha}(\pi).
\end{equation*}
As the left and right expressions in the previous chain of inequalities are equal, we see that all of these terms are equal, which establishes the desired result.
\end{proof}

\newpage 
\section{Proof of Proposition~\ref{prop:vardelta}}
\label{app:vardelta}
\renewcommand{\thefigure}{B.\arabic{figure}}
\renewcommand{\theHfigure}{B.\arabic{figure}}
\renewcommand{\thetable}{B.\arabic{table}}
\renewcommand{\theHtable}{B.\arabic{table}}
\setcounter{figure}{0}
\setcounter{table}{0}

\vardelta*

\begin{proof}
If $\pi = \delta_{\alpha}$ for some $\alpha \in M$, then by definition
$0 \leq \Var(\pi) = \inf_{x \in M} W_2^2(\delta_x, \delta_{\alpha}) \leq  W_2^2(\delta_{\alpha}, \delta_{\alpha}) = 0$, and so $\Var(\pi)=0$.

For the converse, assume $\pi \neq \delta_{\alpha}$ for any $\alpha \in M$, and suppose that $\alpha^*$ satisfies 
$$
\begin{aligned}
\Var(\pi) & = \inf_{\alpha \in M} \int_M d(x,\alpha)^2 \; d\pi(x) \\
& = \int_M d(x,\alpha^*)^2 \; d\pi(x).
\end{aligned}
$$
Since $\pi \neq \delta_{\alpha^*}$, we must have $\pi(M \setminus \{\alpha^*\}) > 0$. Let $U = M \setminus \{\alpha^*\}$, and let 
\begin{equation*}
    U_n = \{ x \in M : d(x,\alpha^*)^2 > 1/n \}.
\end{equation*}
Note that $U = \bigcup_n U_n$, and since $\pi(U) >0$, there exists $n$ such that $\pi(U_n) >0$. Then 
$$
\begin{aligned}
\Var(\pi) & = \int_{M} d(x,\alpha^*)^2 \; d\pi(x) \\ 
& \geq \int_{U_n} d(x,\alpha^*)^2 \; d\pi(x) \\
& \geq \pi(U_n)/n \\
& >0,
\end{aligned}
$$
as desired.
\end{proof}

\newpage

\section{Proof of Proposition~\ref{prop:metric_rescale}}
\label{app:metric_rescale}
\renewcommand{\thefigure}{C.\arabic{figure}}
\renewcommand{\theHfigure}{C.\arabic{figure}}
\renewcommand{\thetable}{C.\arabic{table}}
\renewcommand{\theHtable}{C.\arabic{table}}
\setcounter{figure}{0}
\setcounter{table}{0}

\metricrescale*

\begin{proof}
Observe that the metric topologies defined on $M$ by $d$ and $d_{\lambda}$ are identical and therefore the Borel $\sigma$-algebras are the same and $\mathcal{P}(M;d)=\mathcal{P}(M;d_{\lambda})$. By definition, $\Var_{(M,d_\lambda)}(\pi) = \lambda^2 \Var_{(M,d)}(\pi)$. Therefore, if $\nu = \sup_{\pi' \in \mathcal{P}(M)} \Var_{(M,d)}(\pi')^{1/2}$ is the normalization constant over $(M,d)$, then $\lambda \nu  = \sup_{\pi' \in \mathcal{P}(M)} \Var_{(M,d_\lambda)}(\pi')^{1/2}$ is the normalization constant over $(M,d_\lambda)$ and so $$F_{(M,d_\lambda)}(\pi) = 1- \frac{\Var_{(M,d_\lambda)}(\pi)^{1/2}}{\lambda \nu} = 1-\frac{\lambda \Var_{(M,d)}(\pi)^{1/2}}{\lambda \nu} = F_{(M,d)}(\pi).$$
\end{proof}

\newpage
\section{Proof of Proposition \ref{prop:approximability}}
\label{app:approximability}
\renewcommand{\thefigure}{D.\arabic{figure}}
\renewcommand{\theHfigure}{D.\arabic{figure}}
\renewcommand{\thetable}{D.\arabic{table}}
\renewcommand{\theHtable}{D.\arabic{table}}
\setcounter{figure}{0}
\setcounter{table}{0}

\approximability*

\begin{proof}
 Assume first that $X$ also contains at least two points, so that  both $\nu_{(M,d)}$ and $\nu_{(X,d)}$ are strictly greater than zero. Let $W_2(\pi,\mu)$ be the Wasserstein-$2$ distance between $\pi, \mu \in \mathcal{P}(M)$. 
For any measure $\pi$ on $M$, define
\begin{align*}
    V(\pi) := \Var_{(M,d)}(\pi)^{1/2} = \inf_{{\alpha} \in M} W_2(\pi,\delta_{\alpha}),
\end{align*}
and also let
\begin{align*}
    \nu := \nu_{(M,d)} = \sup_{\pi \in \mathcal{P}(M)} V(\pi).
\end{align*}
Likewise, for any measure $\pi_X$ on $X$, define
\begin{equation*}
    V_X(\pi_X) := \Var_{(X,d)}(\pi_X)^{1/2}  = \inf_{x \in X} W_2(\pi_X,\delta_{x}),
\end{equation*}
and also let
\begin{equation*}
    \nu_X := \nu_{(X,d)} = \sup_{\pi_X \in \mathcal{P}(X)} V_X(\pi_X).
\end{equation*}

    First, note that for any measure $\pi$ on $M$, there is a coupling, $\gamma_0 \in \Gamma[\pi,\pi_X]$ (with marginals equal to $\pi$ and $\pi_X$) supported on the graph of $g$, i.e., the set $\text{graph}(g) = \{(x,g(x)) \in M \times X\; |\; x \in M\}$. Then, since $d(x,g(x)) \leq \epsilon$ for all $x$, we have 
    $$
    \begin{aligned}
    W_2(\pi, \pi_X) & = \left(\inf_{\gamma \in \Gamma[\pi, \pi_X]}\int_{M\times X} d(x,y)^2 \; d\gamma(x,y)\right)^{1/2} \\
    & \leq \left(\int_{M\times X} d(x,y)^2 \; d\gamma_0(x,y)\right)^{1/2} \\
    & = \left(\int_{\text{graph}(g)} d(x,y)^2 \; d\gamma_0(x,y)\right)^{1/2} \\
    & \leq \left(\epsilon^2  \int_{M\times X}\; d\gamma_0(x,y)\right)^{1/2} \\
    & = \epsilon.
    \end{aligned}
    $$
    Then, by the triangle inequality,  for any $\alpha \in M$, and any $\pi \in \mathcal{P}(M)$, we have that
    \begin{equation*}
        W_2(\pi, \delta_{\alpha}) \leq W_2(\pi_X, \delta_{\alpha}) + W_2(\pi, \pi_X) \leq W_2(\pi_X, \delta_{\alpha}) + \epsilon.
    \end{equation*}
    Taking the infimum over all $\alpha$ in $M$, we obtain
    \begin{align} \begin{split} \label{eqn:ineqside11}
        V(\pi) & \leq \inf_{\alpha \in M} W_2(\pi_X, \delta_{\alpha}) + \epsilon \\
        & \leq \inf_{x \in X} W_2(\pi_X, \delta_{x}) + \epsilon \\
        & = V_X(\pi_X) + \epsilon.
        \end{split}
    \end{align}
    Now observe that for any $\alpha$ in $M$ and $\pi$ on $M$, the triangle inequality gives
    \begin{align*}
        W_2(\pi_X, \delta_{g(\alpha)}) & \leq W_2(\delta_{g(\alpha)}, \delta_{\alpha}) + W_2(\pi, \delta_{\alpha}) + W_2(\pi, \pi_X) \\
        & \leq W_2(\pi, \delta_{\alpha}) + 2 \epsilon.
    \end{align*}
    Choosing $\alpha^*$ such that $W_2( \pi, \delta_{\alpha^*}) = V(\pi)$, we see that
    \begin{equation}
    \label{eqn:ineqside12}
        V_X(\pi_X) = \inf_{x \in X} W_2(\pi_X,\delta_x) \leq W_2(\delta_{g(a^*)}, \pi_X) \leq W_2(\delta_{a^*}, \pi) + 2 \epsilon = V(\pi) + 2 \epsilon.
    \end{equation}
    Combining inequalities \eqref{eqn:ineqside11} and \eqref{eqn:ineqside12}, we obtain
    \begin{equation} \label{eqn:FirstBound}
        |V(\pi) - V_X(\pi_X)| \leq 2\epsilon.
    \end{equation}

    Now let $\pi'$ be a measure on $X$. Choose $\alpha^*$ in $M$ such that $V(\pi') = W_2(\pi', \delta_{a^*})$. Then
    \begin{equation*}
        V_X(\pi') = \inf_{x \in X} W_2(\pi', \delta_x) \leq W_2(\pi', \delta_{g(\alpha^*)}) \leq W_2(\pi', \delta_{\alpha^*}) + \epsilon = V(\pi') + \epsilon.
    \end{equation*}
    Taking the supremum over $\pi'$ on $X$, we obtain
    \begin{equation}
    \label{eqn:ineqside21}
        \nu_X \leq \sup_{\pi' \in \mathcal{P}(X)} V(\pi') + \epsilon \leq \sup_{\pi' \in \mathcal{P}(M)} V(\pi') + \epsilon = \nu + \epsilon.
    \end{equation}
    Now taking supremum over all $\pi$ on $M$ in \eqref{eqn:ineqside11}, we also have
    \begin{equation}
    \label{eqn:ineqside22}
        \nu \leq \nu_X + \epsilon.
    \end{equation}
    Combining inequalities \eqref{eqn:ineqside21} and \eqref{eqn:ineqside22}, we see that
    \begin{equation} \label{eqn:SecondBound}
        |\nu - \nu_X| \leq \epsilon.
    \end{equation}
    Finally, by Eq.~\eqref{eqn:FirstBound} and Eq.~\eqref{eqn:SecondBound} we conclude that
    \begin{align*}
        \left|F_{(M,d)}(\pi) - F_{(X,d)}(\pi_X) \right| & =  \left| \frac{1}{\nu} V(\pi) - \frac{1}{\nu_X} V_X(\pi_X) \right| \\
        & =\frac{1}{\nu \cdot \nu_X}\left| \nu_X V(\pi) - \nu V_X(\pi_X) \right| \\
        & \leq  \frac{1}{\nu \cdot \nu_X}\left| \nu_X V(\pi) - \nu_X V_X(\pi_X) \right| \\
         & \quad + \frac{1}{\nu \cdot \nu_X} \left| \nu_X V_X(\pi_X)  - \nu V_X(\nu_X) \right| \\
         & \leq  \frac{1}{\nu}\left| V(\pi) - V_X(\pi_X) \right| + \frac{1}{\nu} \left| \nu_X   - \nu  \right| \\
         & \leq \frac{3\epsilon}{\nu},
    \end{align*}
    which finishes the proof for the case that $X$ contains at least two points.

    If $X = \{s\}$ consists of a single state, then 
    \begin{align*}
        \left|F_{(M,d)}(\pi) - F_{(X,d)}(\pi_X) \right| & = \left|\frac{1}{\nu}V(\pi) - 1 \right| \\
        & =\frac{1}{\nu} \left|V(\pi) - \nu\right| \\ 
        & \leq 1,
    \end{align*}
    since $|\Var(\pi)-\nu| \leq \nu$. But in this case $\epsilon \geq \max_{\alpha \in M} d(\alpha, s)$, and so
    \begin{align*}
    \nu & = \sup_{\pi \in \mathcal{P}(M)}  V(\pi) \\
    & = \sup_{\pi \in \mathcal{P}(M)} \left(\inf_{y \in M} \int_{M} d(x,y)^2 \; d\pi(x) \right)^{1/2} \\
    & \leq \sup_{\pi \in \mathcal{P}(M)} \left(\int_{M} d(x,s)^2 \; d\pi(x) \right)^{1/2} \\
    & \leq \sup_{\pi \in \mathcal{P}(M)} \left( \int_{M} \epsilon^2 \; d\pi(x) \right)^{1/2} \\
    & = \sup_{\pi \in \mathcal{P}(M)} \epsilon \left( \int_{M} \; d\pi(x)  \right)^{1/2} \\
    & = \epsilon.
    \end{align*}
    Therefore $1 \leq \epsilon/\nu < 3\epsilon/\nu$, and the conclusion still holds. 
\end{proof}

\newpage
\section{Proof of Proposition \ref{prop:finite_cyclic_var_max}}
\label{app:finite_cyclic_var_max}
\renewcommand{\thefigure}{E.\arabic{figure}}
\renewcommand{\theHfigure}{E.\arabic{figure}}
\renewcommand{\thetable}{E.\arabic{table}}
\renewcommand{\theHtable}{E.\arabic{table}}
\setcounter{figure}{0}
\setcounter{table}{0}

\begin{lem}
Let $\bm{C}$ be a circulant, $p\times p$ matrix indexed by $0,\ldots,p-1$, so that the $i$-th row, $\bm{C}[i,:]$, is the cyclic shift of $\bm{C}[0,:]$ by $i$ indices to the right, i.e., $\bm{C}[i,j] = \bm{C}[0,(j-i)\mod p]$ for $i,j=0,\ldots, p-1.$ Let $\bm{1} \in \mathbb{R}^p$ be the (column) vector of all 1's. Then for any (column) vector ${\bm v} \in \mathbb{R}^p, \bm{1}^T \bm{C}\bm{v} = C (\bm{1}^T \bm{v}),$ where $C$ is the constant row (and column) sum of $\bm{C}$. In particular, if $\bm{v}$ is a probability vector (its components are non-negative and sum to 1), then $\bm{1}^T \bm{C}\bm{v} = C$ and is independent of $\bm{v}$. 
\label{lem:const_circ_sum}
\end{lem}
\begin{proof}
\begin{align*}
\bm{1}^T \bm{C}\bm{v} & = \sum_{j=0}^{p-1}\left(\sum_{i=0}^{p-1}\bm{C}[i,j]\bm{v}[i]\right)\\
& = \sum_{i=0}^{p-1}\left(\sum_{j=0}^{p-1}\bm{C}[i,j]\bm{v}[i]\right)\\
& = \sum_{i=0}^{p-1}\bm{v}[i]\left(\sum_{j=0}^{p-1}\bm{C}[i,j]\right)\\
& = \sum_{i=0}^{p-1}\bm{v}[i] C \\
& = C\sum_{i=0}^{p-1}\bm{v}[i] \\
& = C (\bm{1}^T \bm{v}). \\
\end{align*}
\end{proof}

\begin{lem}
Assume $\sum_{i=0}^{p-1}a_i = p\lambda$. If there exists $a_j$ such that $a_j \neq \lambda$, then there exists $a_k < \lambda$.
\label{lem:smaller_min}
\end{lem}
\begin{proof}
If $a_j < \lambda$, the statement if true, taking $k=j$. So assume $a_j >\lambda$. If $a_i \geq \lambda$ for all $i\neq j$, then
\begin{align*}
\begin{aligned}
p\lambda & = a_j + \sum_{i\neq j}a_i\\
& \geq a_j + \lambda(p-1)\\
& \implies a_j \leq \lambda,
\end{aligned}
\end{align*}
\noindent a contradiction. Thus, again there exists $a_k < \lambda$.
\end{proof}

\discretecircle*

\begin{proof}
First recall that for any $\pi'\in\mathcal{P}(M)$.
$$\Var(\pi') = \inf_{\alpha \in M} \int_{M}d(x,\alpha)^2 d\pi'(x) = \min_{i\in\{0,\ldots, p-1\}} \sum_{j=0}^{p-1} \bm{C}[i,j]\pi'[j],$$
where $\bm{C}[i,j] = d(i,j)^2= \min(|i-j|/p,1-|i-j|/p)^2$ is the $p\times p$ circulant cost matrix of the optimal transport problem, encoding the squared circular distances between all pairs of states. This follows from the observation that the $i$-th component of the vector $\bm{C}\pi'$ is the Wasserstein-2 distance between $\pi'$ and the delta distribution concentrated on state $i$. Thus, the generalized variance of $\pi'$ is simply the smallest component of the vector $\bm{C}\pi'$. Let $C = \sum_{i=0}^{p-1}\min(i/p,1-i/p)^2$ be the constant row sum of $\bm{C}$. Then, for the uniform measure, $\pi = \frac{1}{p}\bm{1}$, the vector $\bm{C}\pi = \frac{C}{p}\bm{1}$, has constant components (since $\bm{C}$ is circulant), all equal to $C/p$. Therefore $\Var(\pi)=C/p$. 

Let $\pi'$ be any probability distribution on $M$. Then $\sum_{i=0}^{p-1} \bm{C}\pi'[i] = C$, by Lemma~\ref{lem:const_circ_sum}. If there exists $j$  such that $\bm{C}\pi'[j] \neq C/p$, then there must exist $k$ such that $\bm{C}\pi'[k] < C/p$ by Lemma~\ref{lem:smaller_min}. Thus $\Var(\pi') < C/p = \Var(\pi)$ if not all coordinates of $\bm{C}\pi'$ are equal. Now assume that all coordinates of $\bm{C}\pi'$ are equal (to say $\lambda$). Then $C = p\lambda \implies \lambda = C/p$ and therefore $\Var(\pi') = \Var(\pi)$. In either case, $\Var(\pi') \leq \Var(\pi).$
\end{proof}

\begin{rmk}
The condition that guarantees the uniform distribution will have maximal generalized variance in the proof of Proposition \ref{prop:finite_cyclic_var_max} is that the distance matrix for a discrete uniform circle space is circulant. The distance matrix, and thus the Wasserstein-2 cost matrix, for general discrete circle spaces as defined in Example~\ref{exp:discrete_circle_general}, need not be circulant and for this reason the uniform distribution need not maximize asynchrony in such spaces. 

Even for spaces whose distance matrix is circulant, the only guarantee made by Proposition \ref{prop:finite_cyclic_var_max} is that the uniform distribution will maximize asynchrony, not that it is the unique maximizer. For example, consider the 4-state space with circulant Wasserstein-2 cost matrix:
$$
\bm{C} = \begin{bmatrix}
0 & 1 & 2 & 1  \\
1 & 0 & 1 & 2  \\
2 & 1 & 0 & 1  \\
1 & 2 & 1 & 0 
\end{bmatrix}.
$$
Then $\bm{C}\pi = \bm{C} \pi' = [1,1,1,1]$, for
$\pi = [1/4, 1/4, 1/4, 1/4]$ and $\pi' = [1/2, 0, 1/2, 0]$. Thus the maximum generalized variance, which in this case is 1, is not solely achieved by the uniform distribution.

 Uniqueness of the maximizer would be implied if whenever $C \pi$ has constant components, then $\pi$ had constant components. By the preceeding example, we have to appeal to more than circularity of $\bm{C}$ to guarantee this implication. For example, if $\bm{C}$ is circulant and invertible, with row sum $C$, then there will of course be a unique solution to $\bm{C}\pi = [1,\ldots,1]$, given by $\pi = \bm{C}^{-1} [1,\ldots,1] = \bm{1}/C$, the uniform distribution (after normalization). If and why the particular Wasserstein-2 cost matrix for the discrete uniform circular space is invertible is beyond the scope of this paper.

\end{rmk}


\newpage
\section{Proof of Proposition \ref{prop:finite_state_lp_problem}} 
\label{app:finite_state_lp_problem}
\renewcommand{\thefigure}{F.\arabic{figure}}
\renewcommand{\theHfigure}{F.\arabic{figure}}
\renewcommand{\thetable}{F.\arabic{table}}
\renewcommand{\theHtable}{F.\arabic{table}}
\setcounter{figure}{0}
\setcounter{table}{0}

\lpproblem*

\begin{proof}
Let $M=\{0,\ldots,p-1\}$ be a finite metric space with Wasserstein-2 cost matrix $C[i,j] = d(i,j)^2$. 
For each $0 \leq i \leq p-1$, let $f_i:\R^{p} \to \R$ be the linear function $f_i(\pi) = \sum_{j=0}^{p-1}\bm{C}[i,j]\pi[j]$. In other words, if $\pi$ is a probability vector, then $f_i(\pi)$ measures the squared Wasserstein-2 distance between $\pi$ and $\delta_i = [0,\ldots, 1, \ldots, 0]$. Then the optimization problem
    \begin{equation}
    \nu_{(M,d)}^2 = \max_{\pi \in \mathcal{P}(M)} \left\{\min_{0 \leq i \leq p-1} \sum_{j=0}^{p-1}\bm{C}[i,j]\pi[j]\right\}
    \end{equation}
\noindent can be recast as a linear programming problem
    \begin{align}
    \begin{aligned}
        \text{maximize: } & t \\
            \text{subject to: } & t \leq f_{i}(\pi) \\
                & \sum_{i=0}^{p-1} \pi[i] = 1 \\
                & \pi[i] \geq 0.
    \label{eq:norm_const_finite_state_lp}
    \end{aligned}
    \end{align}
    The variable $t$ plays the role of the generalized variance that we are attempting to maximize. Thus, the inequality constraints $t \leq f_{i}(\pi)$ enforce the condition that the generalized variance of a probability distribution (vector) $\pi$ depends only on the \textit{minimal} Wasserstein-2 distance between $\pi$ and any delta distribution. The constraints $\sum_{i=0}^{p-1} \pi[i] = 1$ and $\pi[i] \geq 0$ ensure optimization is done only over probability vectors. \eqref{eq:norm_const_finite_state_lp} can then be written in a more standard form as a minimization LP problem
    \begin{align}
    \begin{aligned}
        \text{minimize  } & c^{\intercal}x \\
        \text{subject to  } & \bm{A_{ub}}x \leq [0,\ldots,0] \\
                & \bm{A_{eq}}x = 1\\
                & 0 \leq x[i], i=0,\ldots,p-1
    \label{eq:norm_const_finite_state_lp_standard}
    \end{aligned}
    \end{align}
    where $x = [\pi[0],\ldots, \pi[p-1], t] \in \mathbb{R}^{p+1}$ is the vector of unknowns, $c = [0,\ldots, 0, -1] \in \mathbb{R}^{p+1}$, $A_{ub} = [-\bm{C};\bm{1}] \in \mathbb{R}^{p \times (p+1)}$ is the negative Wasserstein-2 cost matrix adjoined a column of 1s, and the equality constraint $A_{eq} = [1,\ldots,1,0]$ ensures $\pi$ is a probability vector. 
\end{proof}


\newpage

\section{Squared Wasserstein-2 distance between an empirical distribution and delta distributions on the circle}
\label{app:w22_circle_example}
\renewcommand{\thefigure}{G.\arabic{figure}}
\renewcommand{\theHfigure}{G.\arabic{figure}}
\renewcommand{\thetable}{G.\arabic{table}}
\renewcommand{\theHtable}{G.\arabic{table}}
\setcounter{figure}{0}
\setcounter{table}{0}

\begin{figure}[!ht]
\includegraphics[width=.75\textwidth]{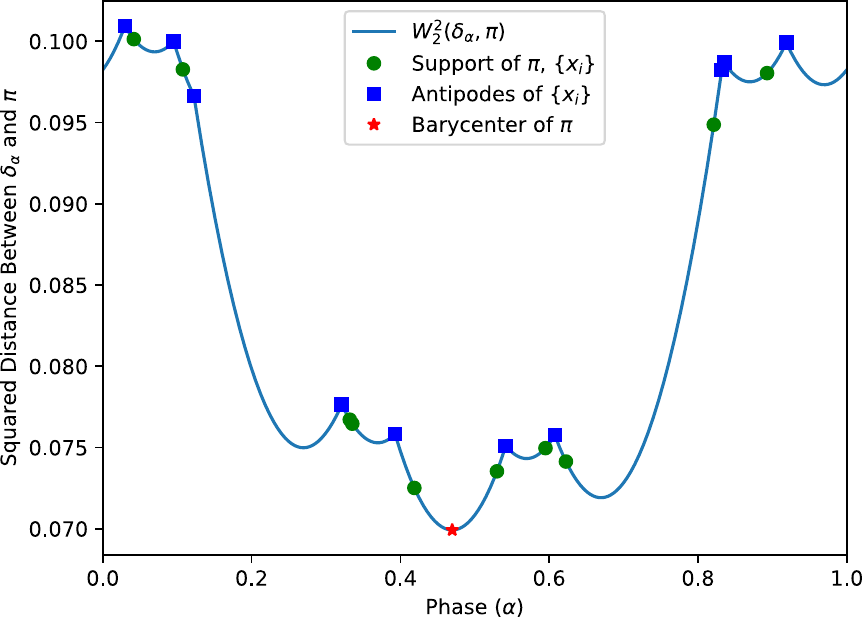}
\caption{The squared Wasserstein-2 distance between an empirical distribution, $\pi = \frac{1}{10}\sum_{i=1}^{10}\delta_{x_i}$, $x_i \in [0,1)$ (blue dots) and $\delta_{\alpha}$ as a function of $\alpha \in [0,1)$. The squared distance is convex between antipodal points $a(x_i) = (x_i+1/2) \mod 1$ (green squares). The barycenter of $\pi$ (red star) is the global minimum of the squared distance function, found by evaluating $W^2_{2}(\pi,\delta_{\alpha_i})$ at candidate barycenters $\{\alpha_i\}$ determined by Algorithm~\ref{alg:empirical_circle_candidate_barycenters}. All points in [0,1) are plotted on the graph of the distance function to aid visualization.}
\label{supp_fig:w22_circle_example}
\end{figure}


\newpage

\section{Proof of Proposition \ref{prop:complexity_empirical_circle}}
\label{app:proof_complexity_empirical_circle}
\renewcommand{\thefigure}{H.\arabic{figure}}
\renewcommand{\theHfigure}{H.\arabic{figure}}
\renewcommand{\thetable}{H.\arabic{table}}
\renewcommand{\theHtable}{H.\arabic{table}}
\renewcommand{\thealgorithm}{H.\arabic{algorithm}}
\renewcommand{\theHalgorithm}{H.\arabic{algorithm}}
\setcounter{figure}{0}
\setcounter{table}{0}
\setcounter{algorithm}{0}

\empiricalcirclealgo*

\begin{proof}
Recall that in this space
\begin{align}
    W_2^2(\delta_{\alpha},\pi) & = \frac{1}{n} \sum_{i=1}^n d(x_i,\alpha)^2 \nonumber \\
    & = \frac{1}{n} \sum_{i=1}^n \min(|x_i-\alpha|, 1-|x_i-\alpha|)^2 \nonumber \\
    & = \frac{1}{n} \sum_{i=1}^n \begin{cases} |x_i-\alpha|^2, & \text{ if } |x_i-\alpha| \leq 1/2 \\ 
    (1-|x_i-\alpha|)^2, & \text{ if } |x_i-\alpha| > 1/2.
    \end{cases}
    \label{appeq:wasserstein2_circle_formula}
\end{align} 
Let the function $f_{x_i}(\alpha) := d(x_i,\alpha)^2$ measure the square Wasserstein-2 cost of moving the delta distribution at $x_i$ to a delta distribution at $\alpha$. To compute the generalized barycenter of $\pi$, we first observe that for each $x_i$, the function $f_{x_i}(\alpha)$ is differentiable on the intervals $[x_i, a(x_i))$ and $(a(x_i), x_i]$, where $a(x_i) = (x_i+1/2) \mod 1$ defines the antipodal point of $x_i$, and intervals are assumed to wrap modulo 1. More precisely, the derivative of the $i$-th term in Eq.~\eqref{appeq:wasserstein2_circle_formula} is 
$$
\frac{df_{x_i}}{d\alpha}(\alpha) = 
\begin{cases}
2(\alpha-x_i), & \text{ if } |x_i-\alpha| < 1/2 \\ 2(\text{sign}(x_i-\alpha)-(x_i-\alpha)), & \text{ if } |x_i-\alpha| > 1/2. 
\end{cases}
$$
By computing the second derivative away from $a(x_i)$, $d^2f_{x_i}/d\alpha^2(\alpha) = 2$, we see that $f_{x_i}$ is convex with respect to $\alpha$ on each $[x_i, a(x_i))$ and $(a(x_i),x_i]$. Thus, on each interval between consecutive antipodal points, $W_2^2(\pi, \delta_{\alpha})$ is also convex, as it is the sum of convex functions. 
    
    Let $x, y, z \in [0,1)$, and assume without loss of generality that $a(x) \leq a(y)$.  If $a(z) \notin (a(x), a(y))$, then $|z-\alpha| < 1/2$ for all $\alpha \in (a(x),a(y))$ or $|z-\alpha| > 1/2$ for all $\alpha \in (a(x),a(y))$. 
    If there exist points $\alpha_{<},\alpha_{>} \in (a(x), a(y))$ such that $|z-\alpha_{<}| < 1/2$ and $|z-\alpha_{>}| > 1/2$, then there must exist $\alpha_{=} \in (a(x), a(y))$ such that $|z-\alpha_{=}| = 1/2$ by the intermediate value theorem. But then $a(z) = \alpha_{=} \in (a(x), a(y))$, a contradiction. Thus the $k$-th term in $dW_2^2(\pi, \delta_{\alpha})/d\alpha$ will be defined either by $2(\text{sign}(x_k-\alpha)-(x_k-\alpha))$ or 
    $2(\alpha-x_k)$ for all $\alpha \in (a(x_i), a(x_j))$ and for every $i,j,k$. Furthermore, if $x_k \in (a(x_i),a(x_j))$, then $|x_k-\alpha| < 1/2$ for all $\alpha \in (a(x_i),a(x_j))$ since otherwise, again, this would imply $a(x_k) \in (a(x_i),a(x_j))$. The same reasoning applies to the intervals $(0,a(x_i))$ and $(a(x_j),1)$, if no other antipodal points are in those intervals. Thus, the functional form of the $k$-th term in the derivative $dW_2^2(\pi, \delta_{\alpha})/d\alpha$ can be determined on each interval in the antipodal-point-partition of [0,1) simply by checking the proximity of $x_k$ to a particular $\alpha \in (a(x_i),a(x_j))$ and the $\text{sign}(x_k-a(x_i)).$ For example, the $k$-th term of the derivative will be
    \setlength{\abovedisplayskip}{5pt}
    \setlength{\belowdisplayskip}{5pt}
    \begin{align*}
    2(\alpha-x_k), & \text{ if } |x_k - (a(x_i)+a(x_j))/2| \leq 1/2 \\
    2(-1+\alpha-x_k), & \text{ if } |x_k - (a(x_i)+a(x_j))/2| > 1/2 \text{ and } x_k \leq a(x_i) \\ 
     2(1+\alpha-x_k), & \text{ if } |x_k - (a(x_i)+a(x_j))/2| > 1/2 \text{ and } x_k > a(x_i).
    \end{align*}
    
    By convexity, a global minimum on each antipodal point interval of $W_2^2(\pi, \delta_{\alpha})$ will either have vanishing derivative or will be an antipodal point. So candidates for the global minimum can be determined by solving  $dW_2^2(\pi, \delta_{\alpha})/d\alpha$=0 for $\alpha$. By the above discussion this can be done efficiently and exactly since $dW_2^2(\pi, \delta_{\alpha})/d\alpha$ is a sum of $n$ terms of the form $2(\alpha-x_k)$, $2(-1+\alpha-x_k)$, or $2(1+\alpha-x_k)$. The equation will thus always have a unique solution of the form $\alpha = x/n$, where $x$ is the sum of all $x_k$ and the difference between the number of points $x_k \leq a(x_i)$ and the number of points $x_k > a(x_i)$, for those $x_k$ whose distance to the interval is not less than 1/2. Note that a critical value may fall outside its partition interval in which it is a proposed minimum. These are not valid candidate solutions and so need not be evaluated as a potential barycenters. In these cases, the minimal value of $W_2^2(\pi, \delta_{\alpha})$ over a partition interval will occur at one of the endpoints of the interval. Pseudo-code of the algorithm determining the candidate barycenters is provided in Algorithm \ref{alg:empirical_circle_candidate_barycenters}.
    
    Thus the exact solution can be found in at most $2n+1$ evaluations of \eqref{eq:wasserstein2_circle_formula}: there are at most $n+1$ valid candidate solutions to $dW_2^2(\pi, \delta_{\alpha})/d\alpha=0$, and at most $n$ antipodal interval end points; $W_2^2(\pi, \delta_{\alpha})$(0) = $W_2^2(\pi, \delta_{\alpha})$(1), so only one of these needs to be checked as a global minimum.  Empirical evidence suggests that for uniformly random data, the expected number of candidate solutions that fall within their defining intervals grows close to $\sqrt{n}$ (Fig.~\ref{supp_fig:valid_candidate_sols}), though the expected total number of evaluations of $W_2^2(\pi,\delta_\alpha)$ remains $\mathcal{O}(n)$.
\end{proof}

\begin{algorithm}[H]
\caption{Computes candidate barycenters of an empirical distribution, $\pi = \sum_{i=1}^{n}\delta_{x_i}/n$ on the circle, $[0,1) \cong S^1.$}
\label{alg:empirical_circle_candidate_barycenters}
\begin{algorithmic}[1]
\Require $\bm{x}$ \Comment{Support of $\pi$: A length-$n$ ARRAY of FLOATS in [0,1).}
\Ensure \textbf{candidateBarycenters}
\State $\bm{ax} \gets [0,1]$
\State $\bm{ax} \gets \bm{ax}.\text{APPEND}(({\bm x} + \frac{1}{2}) \mod 1)$
\State $\bm{ax} \gets \text{SORT}(\text{UNIQUE}(\bm{ax}))$
\Comment{Sorted antipodal point partition of [0,1).}
\State $M \gets \text{length}(\bm{ax}$)
\State $\textbf{candidateBarycenters} \gets \text{ARRAY}[1:M-1]$ of FLOATS
\For{$j \gets 1$ to $M-2$}
        \State $\text{s} \gets \bm{ax}[j]$ \Comment{Partition interval endpoints.}
        \State $\text{e} \gets \bm{ax}[j+1]$
    
    \State $\text{candidate} \gets \text{SUM}(\bm{x})$
    \State $\text{candidate} \gets \text{candidate} + \text{COUNT}(\bm{x} - (s + e)/2 \geq 0.5 \text{ AND } \bm{x} \leq s)$
    \State $\text{candidate} \gets \text{candidate} - \text{COUNT}(\bm{x} - (s+e)/2 \geq 0.5 \text{ AND } \bm{x} > \text{s})$
   \State $\text{candidate} \gets (\text{candidate}/n) \mod 1$
    \If{$\text{candidate} < s \text{ OR } \text{candidate} > e$} \Comment{Invalid candidate.}
        \State $\textbf{candidateBarycenters}[j] \gets \text{NaN}$
    \Else
    \State $\textbf{candidateBarycenters}[j] \gets \text{candidate}$
    \EndIf
\EndFor
$\textbf{candidateBarycenters} \gets \textbf{candidateBarycenters}.\text{APPEND}(\bm{ax})$
\end{algorithmic}
\end{algorithm}


\newpage

\section{Empirical growth of the number of valid candidate barycenters of uniformly random empirical distributions on the circle}
\label{app:valid_cand_barycenters}
\renewcommand{\thefigure}{I.\arabic{figure}}
\renewcommand{\theHfigure}{I.\arabic{figure}}
\renewcommand{\thetable}{I.\arabic{table}}
\renewcommand{\theHtable}{I.\arabic{table}}
\setcounter{figure}{0}
\setcounter{table}{0}

\begin{figure}[!ht]
\includegraphics[width=.75\textwidth]{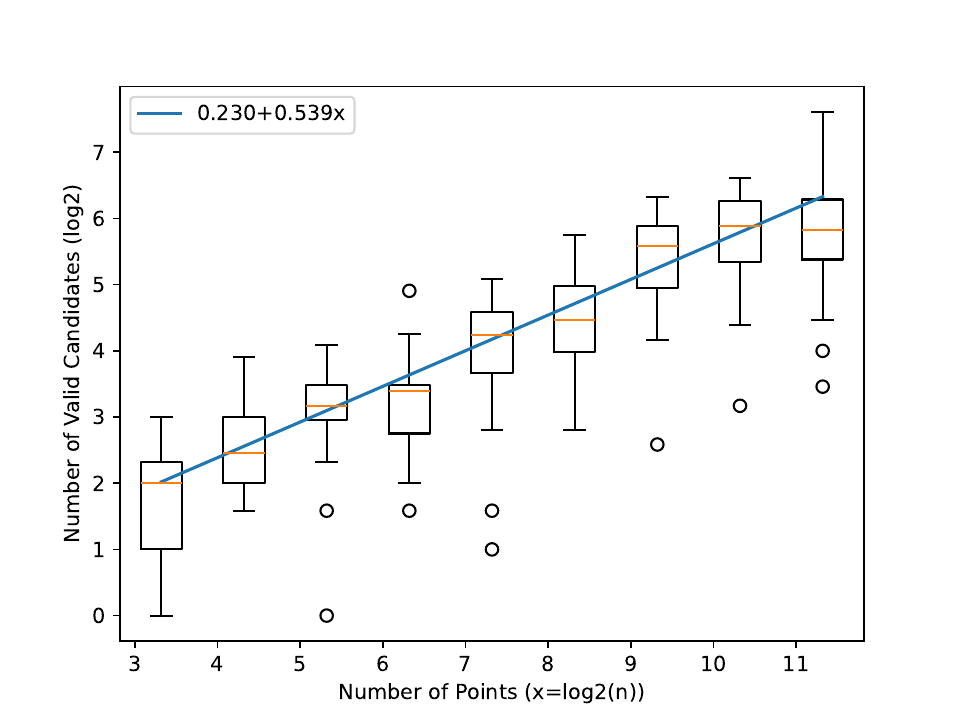}
\caption{Empirical growth of the number of valid candidate barycenters of uniformly random empirical distributions on the circle.}
\label{supp_fig:valid_candidate_sols}
\end{figure}


\newpage
\section{\textit{P. falciparum} strain-specific morphological stage state spaces}
\label{supp:pfalc_stage_spaces}
\renewcommand{\thefigure}{J.\arabic{figure}}
\renewcommand{\theHfigure}{J.\arabic{figure}}
\renewcommand{\thetable}{J.\arabic{table}}
\renewcommand{\theHtable}{J.\arabic{table}}
\setcounter{figure}{0}
\setcounter{table}{0}

\begin{table}[!ht]
\centering
{\setlength{\tabcolsep}{10pt}%
\begin{tabular}{p{1.3cm}p{2.4cm}p{2.4cm}p{1.0cm}}\toprule
\textit{P. falc.} Strain & Ring-Troph. Threshold & Troph.-Schiz. Threshold & Period (hours)\\
\midrule
3D7 & 0.312 & 0.602 & 39 \\
D6 & 0.478 & 0.629 & 36 \\
FVO & 0.571 & 0.797 & 43 \\
SA250 & 0.419 & 0.680 & 54 \\ 
HB3 & 0.390  & 0.657 & 46 \\
\bottomrule
\end{tabular}}
\caption{Model-inferred phases determining transitions between \textit{P. falciparum} discrete morphological stages and strain-specific developmental cycle periods, as reported in \cite{science}.}
\label{supp_tab:pfalc_stage_thresholds}
\end{table}

The thresholds in Table~\ref{supp_tab:pfalc_stage_thresholds} model strain-specific subintervals encoding IEC phases (fraction of the IEC) corresponding to ring stage, $[0, \theta^{\text{strain}}_{\text{rt}})$, trophozoite stage $[\theta^{\text{strain}}_{\text{rt}}, \theta^{\text{strain}}_{\text{ts}})$, and schizont stage $[\theta^{\text{strain}}_{\text{ts}}, 1)$. Bursting and reinvasion is assumed to occur at phase 0, and phases are taken to be in [0,1). So, for example, a 3D7 parasite with an IEC phase in [0,0.312) is assumed to be, morphologically, in ring stage, while an HB3 parasite with an cycle phase in [0.390,0.657) is in the trophozoite stage. These stage intervals were further refined into equal-length early/late, or early/middle/late sub-stages to define state spaces with 6 or 9 states respectively for each strain. Synchronies over time as measured in the circle and in each discretization of the circle are shown in Fig.~\ref{fig:science_reanalysis2}. The distance matrices $D_{\text{strain},6}$ and $D_{\text{strain},9}$ for these discrete circular state spaces, as defined in Example \ref{exp:discrete_circle_general}, are 


\[
D_{\text{3D7},6}= \begin{bNiceMatrix}[first-col,first-row]   & 0 & 1 & 2 & 3 & 4 & 5\\
 0 & 0 & 0.156 & 0.306 & 0.451 & 0.377 & 0.177\\ 
 1 & 0.156 & 0 & 0.150 & 0.295 & 0.467 & 0.334\\ 
 2 & 0.306 & 0.150 & 0 & 0.145 & 0.317 & 0.484\\ 
 3 & 0.451 & 0.295 & 0.145 & 0 & 0.172 & 0.371\\ 
 4 & 0.376 & 0.467 & 0.317 & 0.172 & 0 & 0.199\\ 
 5 & 0.177 & 0.334 & 0.484 & 0.371 & 0.199 & 0 
\end{bNiceMatrix}
\]

\[
D_{\text{3D7},9} = \begin{bNiceMatrix}[first-col,first-row]   & 0 & 1 & 2 & 3 & 4 & 5 & 6 & 7 & 8\\
 0 & 0 & 0.104 & 0.208 & 0.308 & 0.405 & 0.498 & 0.384 & 0.251 & 0.118\\
 1 & 0.104 & 0 & 0.104 & 0.204 & 0.301 & 0.398 & 0.488 & 0.355 & 0.222\\ 
 2 & 0.208 & 0.104 & 0 & 0.100 & 0.197 & 0.294 & 0.408 & 0.459 & 0.326\\ 
 3 & 0.308 & 0.204 & 0.100 & 0 & 0.097 & 0.193 & 0.308 & 0.441 & 0.427\\ 
 4 & 0.405 & 0.301 & 0.197 & 0.097 & 0 & 0.097 & 0.211 & 0.344 & 0.477\\ 
 5 & 0.498 & 0.398 & 0.294 & 0.193 & 0.097 & 0 & 0.115 & 0.247 & 0.380\\ 
 6 & 0.384 & 0.488 & 0.408 & 0.308 & 0.211 & 0.115 & 0 & 0.133 & 0.265\\ 
 7 & 0.251 & 0.355 & 0.459 & 0.441 & 0.344 & 0.247 & 0.133 & 0 & 0.133\\ 
 8 & 0.118 & 0.222 & 0.326 & 0.427 & 0.477 & 0.380 & 0.265 & 0.133 & 0 
\end{bNiceMatrix}
\]

\[
D_{\text{D6},6} = \begin{bNiceMatrix}[first-col,first-row]   & 0 & 1 & 2 & 3 & 4 & 5\\
 0 & 0 & 0.239 & 0.396 & 0.472 & 0.398 & 0.212\\ 
 1 & 0.239 & 0 & 0.157 & 0.233 & 0.363 & 0.451\\ 
 2 & 0.396 & 0.157 & 0 & 0.076 & 0.206 & 0.392\\ 
 3 & 0.472 & 0.233 & 0.076 & 0 & 0.131 & 0.316\\ 
 4 & 0.398 & 0.363 & 0.206 & 0.131 & 0 & 0.185\\ 
 5 & 0.212 & 0.451 & 0.392 & 0.316 & 0.185 & 0\\ 
\end{bNiceMatrix}
\]

\[ 
D_{\text{D6},9} = \begin{bNiceMatrix}[first-col,first-row]   & 0 & 1 & 2 & 3 & 4 & 5 & 6 & 7 & 8\\ 
 0 & 0 & 0.159 & 0.319 & 0.423 & 0.474 & 0.476 & 0.389 & 0.265 & 0.141\\
 1 & 0.159 & 0 & 0.159 & 0.264 & 0.315 & 0.365 & 0.452 & 0.424 & 0.301\\ 
 2 & 0.319 & 0.159 & 0 & 0.105 & 0.155 & 0.206 & 0.293 & 0.416 & 0.460\\ 
 3 & 0.423 & 0.264 & 0.105 & 0 & 0.050 & 0.101 & 0.188 & 0.311 & 0.435\\ 
 4 & 0.474 & 0.315 & 0.155 & 0.050 & 0 & 0.050 & 0.137 & 0.261 & 0.385\\ 
 5 & 0.476 & 0.365 & 0.206 & 0.101 & 0.050 & 0 & 0.087 & 0.211 & 0.334\\ 
 6 & 0.389 & 0.452 & 0.293 & 0.188 & 0.137 & 0.087 & 0 & 0.124 & 0.247\\ 
 7 & 0.265 & 0.424 & 0.416 & 0.311 & 0.261 & 0.211 & 0.124 & 0 & 0.124\\ 
 8 & 0.142 & 0.301 & 0.460 & 0.435 & 0.385 & 0.334 & 0.247 & 0.124 & 0
\end{bNiceMatrix}
\]

\[
D_{\text{FVO},6} = \begin{bNiceMatrix}[first-col,first-row]   & 0 & 1 & 2 & 3 & 4 & 5\\ 
 0 & 0 & 0.285 & 0.485 & 0.402 & 0.295 & 0.194\\ 
 1 & 0.285 & 0 & 0.199 & 0.312 & 0.420 & 0.479\\ 
 2 & 0.485 & 0.199 & 0 & 0.113 & 0.220 & 0.322\\ 
 3 & 0.402 & 0.312 & 0.113 & 0 & 0.107 & 0.209\\ 
 4 & 0.295 & 0.420 & 0.220 & 0.107 & 0 & 0.101\\ 
 5 & 0.194 & 0.479 & 0.322 & 0.209 & 0.102 & 0
\end{bNiceMatrix}
\]

\[
D_{\text{FVO},9} = \begin{bNiceMatrix}[first-col,first-row]   & 0 & 1 & 2 & 3 & 4 & 5 & 6 & 7 & 8\\
 0 & 0 & 0.190 & 0.381 & 0.487 & 0.411 & 0.336 & 0.264 & 0.197 & 0.129\\ 
 1 & 0.190 & 0 & 0.190 & 0.323 & 0.398 & 0.474 & 0.455 & 0.387 & 0.319\\ 
 2 & 0.381 & 0.190 & 0 & 0.133 & 0.208 & 0.284 & 0.355 & 0.423 & 0.490\\ 
 3 & 0.486 & 0.323 & 0.133 & 0 & 0.075 & 0.151 & 0.222 & 0.290 & 0.358\\ 
 4 & 0.411 & 0.398 & 0.208 & 0.075 & 0 & 0.075 & 0.147 & 0.215 & 0.282\\ 
 5 & 0.336 & 0.474 & 0.284 & 0.151 & 0.075 & 0 & 0.072 & 0.139 & 0.207\\ 
 6 & 0.264 & 0.455 & 0.355 & 0.222 & 0.147 & 0.072 & 0 & 0.068 & 0.135\\ 
 7 & 0.197 & 0.387 & 0.423 & 0.290 & 0.215 & 0.139 & 0.068 & 0 & 0.068\\ 
 8 & 0.129 & 0.319 & 0.490 & 0.358 & 0.282 & 0.207 & 0.135 & 0.068 & 0 
\end{bNiceMatrix}
\]

\[
D_{\text{SA250},6} = \begin{bNiceMatrix}[first-col,first-row]   & 0 & 1 & 2 & 3 & 4 & 5\\
 0 & 0 & 0.209 & 0.380 & 0.490 & 0.345 & 0.185\\
 1 & 0.209 & 0 & 0.170 & 0.300 & 0.446 & 0.394\\
 2 & 0.379 & 0.170 & 0 & 0.131 & 0.276 & 0.436\\
 3 & 0.490 & 0.300 & 0.131 & 0 & 0.145 & 0.305\\
 4 & 0.345 & 0.446 & 0.276 & 0.145 & 0 & 0.160\\
 5 & 0.185 & 0.394 & 0.436 & 0.305 & 0.160 & 0
\end{bNiceMatrix}
\]

\[
D_{\text{SA250},9} = \begin{bNiceMatrix}[first-col,first-row]   & 0 & 1 & 2 & 3 & 4 & 5 & 6 & 7 & 8\\
 0 & 0 & 0.140 & 0.279 & 0.393 & 0.480 & 0.433 & 0.337 & 0.230 & 0.123\\
 1 & 0.140 & 0 & 0.140 & 0.253 & 0.340 & 0.427 & 0.476 & 0.369 & 0.263\\
 2 & 0.279 & 0.140 & 0 & 0.113 & 0.200 & 0.287 & 0.384 & 0.491 & 0.402\\
 3 & 0.393 & 0.253 & 0.113 & 0 & 0.087 & 0.174 & 0.271 & 0.378 & 0.484\\
 4 & 0.480 & 0.340 & 0.200 & 0.087 & 0 & 0.087 & 0.184 & 0.291 & 0.397\\
 5 & 0.433 & 0.427 & 0.287 & 0.174 & 0.087 & 0 & 0.097 & 0.204 & 0.310\\
 6 & 0.336 & 0.476 & 0.384 & 0.271 & 0.184 & 0.097 & 0 & 0.107 & 0.213\\
 7 & 0.230 & 0.369 & 0.491 & 0.378 & 0.290 & 0.204 & 0.107 & 0 & 0.107\\
 8 & 0.123 & 0.263 & 0.403 & 0.484 & 0.397 & 0.310 & 0.213 & 0.107 & 0
\end{bNiceMatrix}
\]

\[
D_{\text{HB3},6} = \begin{bNiceMatrix}[first-col,first-row]   
& 0 & 1 & 2 & 3 & 4 & 5\\
 0 & 0 & 0.195 & 0.359 & 0.493 & 0.355 & 0.183\\
 1 & 0.195 & 0 & 0.164 & 0.298 & 0.450 & 0.378\\
 2 & 0.359 & 0.164 & 0 & 0.134 & 0.286 & 0.457\\
 3 & 0.493 & 0.298 & 0.134 & 0 & 0.152 & 0.324\\
 4 & 0.355 & 0.450 & 0.286 & 0.152 & 0 & 0.171\\
 5 & 0.183 & 0.378 & 0.457 & 0.324 & 0.171 & 0
\end{bNiceMatrix}
\]

\[
D_{\text{HB3},9} = \begin{bNiceMatrix}[first-col,first-row]   
& 0 & 1 & 2 & 3 & 4 & 5 & 6 & 7 & 8\\
 0 & 0 & 0.130 & 0.260 & 0.369 & 0.459 & 0.452 & 0.351 & 0.236 & 0.122\\
 1 & 0.130 & 0 & 0.130 & 0.239 & 0.329 & 0.418 & 0.481 & 0.367 & 0.252\\
 2 & 0.260 & 0.130 & 0 & 0.110 & 0.199 & 0.288 & 0.389 & 0.496 & 0.382\\
 3 & 0.369 & 0.239 & 0.110 & 0 & 0.089 & 0.178 & 0.280 & 0.394 & 0.492\\
 4 & 0.459 & 0.329 & 0.199 & 0.089 & 0 & 0.089 & 0.191 & 0.305 & 0.419\\
 5 & 0.452 & 0.418 & 0.288 & 0.178 & 0.089 & 0 & 0.102 & 0.216 & 0.330\\
 6 & 0.351 & 0.481 & 0.389 & 0.280 & 0.191 & 0.102 & 0 & 0.114 & 0.229\\
 7 & 0.236 & 0.366 & 0.496 & 0.394 & 0.305 & 0.216 & 0.114 & 0 & 0.114\\
 8 & 0.122 & 0.252 & 0.382 & 0.492 & 0.419 & 0.330 & 0.229 & 0.114 & 0
\end{bNiceMatrix}
\]

\end{document}